\documentclass{article}
\usepackage[utf8]{inputenc}
\usepackage{setspace}
\usepackage{comment}

\usepackage{geometry}
\usepackage{graphicx,wrapfig,lipsum,booktabs}
\graphicspath{ {images/} }
\usepackage{amsmath}
\usepackage{amsthm}
\usepackage{tikz}
\usepackage{pgfplots}
\pgfplotsset{compat = newest, width = 10cm}
\usepgfplotslibrary{fillbetween}
\usepackage{caption}
\usepackage{amssymb}
\usepackage{mathtools}
\usepackage{cite}
\usepackage{bbm}
\allowdisplaybreaks
\def\ER{Erd\H{o}s-R\'{e}nyi }

\def\E{\mathbb{E}}

\def\N{\mathbb{N}}

\def\E{\mathbb{E}}
\def\R{\mathbb{R}}

\def\eps{\varepsilon}
\def\del{\delta}

\def\1{\mathbf{1}}

\def\tce{t_c + \eps}
\def\tce2{t_c + \frac{\eps}{2}}
\def\ER{Erd\H{o}s-R\'{e}nyi }

\def\Gdb{G_{n,d}^{\mathrm{bip}}}

\def\Xg{X_{\gamma}}

\def\bT{\mathbf{T}}

\input{insbox}

\newtheorem{Theorem}{Theorem}
\newtheorem{Proposition}[Theorem]{Proposition}

\newtheorem{Corollary}[Theorem]{Corollary}
\newtheorem{lemma}[Theorem]{Lemma}
\newtheorem{Conjecture}{Conjecture}

\newtheorem*{observation}{Observation}
\newtheorem{Definition}[Theorem]{Definition}

\title{On the hardness of finding balanced independent sets in random bipartite graphs}
\author{Will Perkins\thanks{School of Computer Science, Georgia Institute of Technology; supported in part by NSF grant CCF-2309708.}, Yuzhou Wang\thanks{School of Mathematics, Georgia Institute of Technology; supported in part by a Georgia Tech ARC Fellowship.}}

\date{\today}

\begin{document}

\maketitle

\begin{abstract}
 We consider the algorithmic problem of finding large \textit{balanced} independent sets in sparse random bipartite graphs, and more generally the problem of finding independent sets with specified proportions of vertices on each side of the bipartition.  In a bipartite graph it is trivial to find an independent set of density at least half (take one of the partition classes). In contrast, in a random bipartite graph of average degree $d$, the largest balanced independent sets (containing equal number of vertices from each class) are typically of density $(2+o_d(1)) \frac{\log d}{d}$.  Can we find such  large balanced independent sets in these graphs efficiently?  By utilizing the overlap gap property  and the low-degree algorithmic framework, we prove that local and low-degree algorithms (even those that know the bipartition) cannot find balanced independent sets of density greater than $(1+\eps) \frac{\log d}{d}$ for any $\eps>0$ fixed and $d$ large but constant.  This factor $2$ \textit{statistical--computational gap} between what exists and what local algorithms can achieve is analogous to the gap for finding large independent sets in (non-bipartite) random graphs.  Our results therefor suggest that this gap is pervasive in many models, and that hard computational problems can lurk inside otherwise tractable ones.   A particularly striking aspect of the gap in bipartite graphs is that the algorithm achieving the lower bound is extremely simple and can be implemented as a $1$-local algorithm and a degree-$1$ polynomial (a linear  function).

 More generally, we provide a tight characterization of the power of local and low-degree algorithms to find $\gamma$-balanced independent sets in random $d$-regular bipartite graphs (with $\gamma \le 1/2$ proportion of vertices on one side of the partition): for large $d$, local algorithms can find $\gamma$-balanced independent sets a factor  $(1-\gamma)$ smaller than those that exist whp in random bipartite graphs, and no larger.
\end{abstract}

\section{Introduction}

    It is well known that finding large independent sets in graphs is a  hard computational problem; not only is finding a maximum-size independent set in a graph NP-hard, it is also NP-hard to approximate the maximum size to within a factor $n^{1- \epsilon}$~\cite{hastad1996clique,khot2001improved}.  This motivates the question of what hard instances of max independent set look like; or alternatively, whether finding large independent sets in `typical' instances is tractable or not.

    This leads to one of the most striking and well studied problems in average-case complexity: the problem of finding large cliques or independent sets in random graphs. There is a notorious gap between what is known existentially and what can be found by known efficient algorithms. In the \ER random graph $G(n,1/2)$ the largest independent set (or clique) is typically of size $\sim 2 \log_2 n$ while the best known polynomial-time algorithms (in fact very simple greedy algorithms) can only find independent sets of size $\sim \log_2 n$ whp (with probability $1-o(1)$ over the instance).  In 1976 Karp proposed the conjecture  that for all $\eps>0$ fixed, no polynomial-time algorithm can find an independent set of size $(1+\eps) \log_2 n$ whp~\cite{karp1976probabilistic}; and since then, this factor 2 gap between the existential and algorithmic results has resisted all attacks; in modern language, the problem exhibits a \textit{statistical--computational gap}.  Attempts to understand this problem include the introduction of the planted clique problem by Jerrum~\cite{jerrum1992large}.

The factor 2 statistical--computational gap for independent sets extends to  sparse \ER random graphs, $G(n,d/n)$,  for $d$ fixed but large.  Frieze~\cite{frieze1990independence} showed that whp the density of the largest independent set in these graphs is $(2+o_d(1))\log d/d$ as $d \to \infty$. Yet no known polynomial-time algorithms can find an independent set of density  $(1+\eps) \log d/d$ whp for any $\eps>0$ fixed independent of $d$, and it is conjectured that there is no polynomial-time algorithm that achieves this.  

Proving this conjecture unconditionally would amount to proving a stronger statement than $\text{P} \ne \text{NP}$, an ambitious goal to say the least.  Instead, researchers have focused on providing rigorous evidence of computational intractability in various forms.  In the broader studying of average-case complexity and statistical--computational gaps, such evidence includes reductions between problems (often starting with planted clique)~\cite{berthet2013complexity,brennan2018reducibility,brennan2019optimal}; slow-mixing results for Markov chains~\cite{jerrum1992large,chen2023almost}; integrality gaps for convex relaxations~\cite{ma2015sum,hopkins2017power,barak2019nearly,jones2022sum}; and impossibility results for restricted classes of algorithms, including statistical query algorithms~\cite{feldman2017statistical,feldman2018complexity}, low-degree algorithms~\cite{hopkins2017efficient,gamarnik2020low,schramm2022computational,wein2022optimal} and local or stable algorithms~\cite{gamarnik2017limits,gamarnik2017performance,rahman2017local}.  

Proving any rigorous hardness results for finding large independent sets in sparse random graphs is challenging for two reasons.  First,  methods that lose small polynomial or even logarithmic factors  are ineffective in understanding   a constant factor  statistical--computational gap.  Second, many hardness techniques work in the setting of `planted problems' in which a good solution is planted and obscured by random noise; this includes planted clique, sparse PCA, and community detection.  However, the framework of the `Overlap Gap Property' (OGP), as pioneered by Gamarnik and Sudan~\cite{gamarnik2017limits} is effective in this setting and has led to tight lower bounds for both local algorithms~\cite{rahman2017local} and low-degree polynomials~\cite{wein2022optimal}.

    Our main question is on the universality of this statistical--computational gap for finding large independent sets.   In what other random graph models (and for what kind of independent sets) do we see such a gap?

    We start by considering a very `easy' independent set problem: finding large independent sets in bipartite graphs.  Finding a maximum size independent set in a bipartite graph has a polynomial-time algorithm based on solving a max-flow problem, and in a graph with some additional structure the problem can be trivial (take one of the parts of the bipartition).    The natural model of a sparse random  bipartite graph is a bipartite \ER random graph:  $\Gdb$ is the random bipartite graph on vertex sets $(L,R)$, each of size $n$, with  every possible edge  in $L \times R$ included independently with probability $p = d/n$.   Locally, $\Gdb$ looks like the  sparse \ER random graph $G(n,d/n)$: the local neighborhood structure of a vertex converges to a Poisson$(d)$ Galton-Watson tree, but globally these distributions are very different:  the largest independent set in $\Gdb$ typically   has density  $1/2$ as opposed to $(2+o_d(1)) \log d/d$ for $G(n,d/n)$.  Note that we define $\Gdb$ on a set of $2n$ vertices while $G(n,d/n)$ is defined on a set of $n$ vertices; we will refer to the sizes of independent sets in terms of their density (fraction of vertices occupied) to prevent any confusion in comparing the models.  
    
What can we say from a probabilistic or algorithmic perspective about the independent sets of $\Gdb$? Finding very large independent sets in $\Gdb$ efficiently is trivial, and in fact  it is even possible to approximately count independent sets in such graphs efficiently and sample    them uniformly (or with sufficiently large exponential weights)  ~\cite{jenssen2020algorithms,liao2019counting,chen2022sampling}.  In a sense the situation is as nice as possible, with little hint of computational hardness. However, there is a powerful idea from statistical physics that suggests that imposing a global constraint in this model can induce `glassy' equilibrium behavior~\cite{mezard1987mean} and perhaps introduce computational intractability.  

A natural global constraint in the context of independent sets in bipartite graphs is requiring them to be \textit{balanced}: with an equal number of vertices from each side of the bipartition.  Balanced independent sets arise in several different contexts in graph theory and combinatorics~\cite{axenovich2021bipartite,barber2012note,ramras2010balanced,park2022note,chakraborti2021extremal} and in statistical physics and optimization they play a similar role in the context of independent sets that bisections play in the context of cuts or the Ising model~\cite{zdeborova2010conjecture,alaoui2021local}. 
 In fact, just as finding a min bisection is NP-hard while finding a min cut is in P, finding a maximum balanced independent set in a bipartite graph is NP-hard~\cite{garey1979computers,feige2002relations,feige2004hardness} while finding a maximum independent set in a bipartite graph is in P.     

In Theorem~\ref{Thm:balexist} below, we show that the largest balanced independent set in $\Gdb$ is of density $(2+o_d(1)) \frac{\log d}{d}$ whp.  Moreover, there are simple, efficient algorithms that can find balanced independent sets of density $(1+o_d(1)) \frac{\log d}{d}$ in $\Gdb$~\cite{lauer2007large,rahman2017local} and even in the much more general setting of bipartite graphs of average degree at most $d$~\cite{chakraborti2021extremal}.  Can one do any better in $\Gdb$, or is there a statistical--computational gap?

We conjecture that this problem is indeed computationally intractable.  
    \begin{Conjecture}
    \label{conjHardness}
        For every $\eps>0$ there is $d$ large enough so that there is no polynomial-time algorithm that whp finds  a balanced independent set of density $ (1+ \eps) \frac{\log d}{d}$ in $\Gdb$.  
    \end{Conjecture}

The main result of this paper is  rigorous evidence towards Conjecture~\ref{conjHardness}: that both \textit{local algorithms} and \textit{low-degree algorithms} cannot find such large balanced independent sets, but can find balanced independent sets of density $(1+o_d(1)) \log d/d$; that is, such algorithms are only half optimal.  In the following sections we formally define local and low-degree algorithms for finding independent sets in bipartite graphs; in particular, in defining local algorithms we allow each vertex knowledge of which partition it lies in.   But first we state our main results informally.

\begin{Theorem}[informal]
For every $\eps>0$ there is $d$ large enough so that there are no local or low-degree algorithms that whp find  a balanced independent set of density $ (1+ \eps) \frac{\log d}{d}$ in $\Gdb$.  
\end{Theorem}

More generally, we could ask for algorithms to find $\gamma$-balanced independent sets: independent sets with $\gamma$ proportion of their vertices in $L$ and $(1-\gamma)$ proportion in $R$.  

\begin{Definition}
    Given a bipartite graph $G$ with a specified bipartition $(L,R)$, a $\gamma$-balanced independent set $I$ is an independent set of $G$ so that  $ \big | | I \cap L| - \gamma | I|  \big| < 1$.
\end{Definition}
In particular a $\frac{1}{2}$-balanced independent set is a (nearly) balanced independent set.   We will assume throughout WLOG that $\gamma \in [0,1/2]$.  Our next result is a tight characterization of the performance of local and low-degree algorithms for finding $\gamma$-balanced independent sets: there is a factor $1/(1-\gamma)$ statistical--computational gap.

\begin{Theorem}[informal]
   Fix $\gamma \in (0,1/2]$.  The following hold.
   \begin{itemize}
       \item Whp the largest $\gamma$-balanced independent set in $\Gdb$ has density $\left( \frac{1}{2 \gamma (1- \gamma) } +o_d(1) \right) \frac{\log d}{d}$.

\item There are local and low-degree algorithms that whp find $\gamma$-balanced independent sets of density $\left( \frac{1}{2 \gamma } +o_d(1) \right) \frac{\log d}{d}$.

   \item For every $\eps>0$ there is $d$ large enough so that there are no local or low-degree algorithms that whp find  a $\gamma$-balanced independent set of density $\left( \frac{1}{2 \gamma } +\eps \right) \frac{\log d}{d}$ in $\Gdb$.
   \end{itemize}
\end{Theorem}

We can visualize all the main results in the  diagram below.   The diagram indicates if independents sets with density $x \cdot \log d/d$ in $L$ and $y \cdot \log d/d$ in $R$ typically exist in $\Gdb$ (for large $d$) and if so, whether finding them is tractable or intractable for local and low-degree algorithms.  As indicated, if either $x$ or $y$ is strictly below $1$, then the other can be arbitrarily large and such independent sets can be found with local algorithms once $d$ is large enough.  On the other hand, as soon as both $x$ and $y$ are strictly bigger than $1$ then finding such independent sets is hard for local and low-degree algorithms.  The diagonal of the diagram indicates the factor $2$ gap for balanced independent sets.
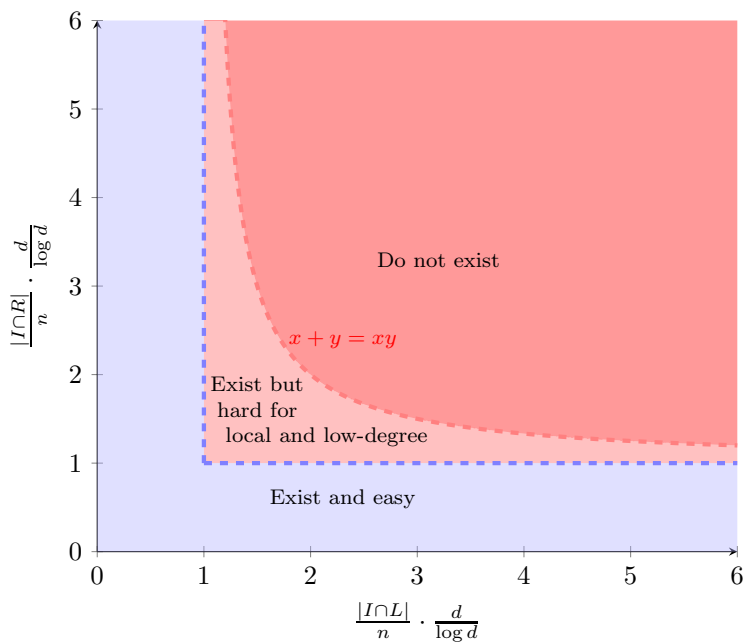
\begin{center}
\begin{tikzpicture}
\label{figPhase}

   \begin{axis}[
   	xmin=0, xmax=6,
   	ymin=0, ymax=6,
   	xtick distance=1, ytick distance=1,
        axis lines=left,
        xlabel = $\frac{|I \cap L|}{n} \cdot \frac{d}{\log d}$, 
        ylabel = $\frac{|I \cap R|}{n} \cdot \frac{d}{\log d}$,
     ]

    \addplot [domain=1:6, samples=1000, name path=f, ultra thick, dashed, color=red!50]
        {x/(x-1)};

    \addplot [domain=1:6, samples=100, name path=g, ultra thick, dashed, color=blue!50]
        {1};

    \addplot[domain = 1:6, mark=none, name path=h, ultra thick, dashed, color=blue!50] (1,x);

    \path [name path=xaxis]
      (\pgfkeysvalueof{/pgfplots/xmin},0) --
      (\pgfkeysvalueof{/pgfplots/xmax},0);

    \plot[name path=y_is_6,thick,opacity=0,samples=100,domain=0:6] {6};

    \plot[name path=y_is_1,thick,opacity=0,samples=100,domain=1:6] {1};
      
    \addplot[red!60, opacity=0.4] fill between[of=f and g, soft clip={domain=1:6}];

    \addplot[red!100, opacity=0.4] fill between[of=y_is_6 and f, soft clip={domain=1:6}];
    
    \addplot[blue!30, opacity=0.4] fill between [of= y_is_6 and xaxis, soft clip={domain=0:1}];

    \addplot[blue!30, opacity=0.4] fill between [of= y_is_1 and xaxis, soft clip={domain=1:6}];
    
    \node[color=red, font=\footnotesize] at (axis cs: 2.3,2.4) {$x+y = xy$};

    \node[color=black, font=\footnotesize] at (axis cs: 3.2,3.3) {Do not exist};

    \node[color=black, font=\footnotesize] at (axis cs: 1.5,1.9) {Exist but};
    \node[color=black, font=\footnotesize] at (axis cs: 1.5,1.6) {hard for};
    \node[color=black, font=\footnotesize] at (axis cs: 2.15,1.3) {local and low-degree};

    \node[color=black, font=\footnotesize] at (axis cs: 2.3,0.6) {Exist and easy};

   \end{axis}

  \end{tikzpicture}
\end{center}

In fact, there is  one algorithmic result and one hardness result that lead to the diagram above:
\begin{itemize}
    \item The algorithm is from~\cite{chakraborti2021extremal} and is perhaps the most naive algorithm one could devise to find balanced independent sets in a bipartite graph: include each vertex in $L$ independently with probability $p$ and then include each vertex in $R$ that is not blocked by an occupied vertex in $L$.  We show below that this algorithm can be implemented as a $1$-local algorithm and as a degree-$1$ polynomial, and by setting $p = (1-\eps)\log d/d$, the algorithm achieves densities $(1-\eps)\log d/d$ in $L$ and $d^{\eps-1}$ in $R$; this accounts for the entire blue `easy' region in Figure~\ref{figPhase}. 
    \item Using the Overlap Gap Property we show that no local or low-degree algorithms can find an independent set of density at least $(1+\eps)\log d/d$ in both $L$ and $R$ in $\Gdb$ (for any $\eps>0$ and $d$ large enough).  This accounts for the entire `hard' region in Figure~\ref{figPhase}.
\end{itemize}

Perhaps the most striking aspect of the results is that a completely trivial algorithm achieves the best possible performance (to first order) among all local and low-degree algorithms (and this algorithm can be implemented as a $1$-local algorithm and degree-$1$ polynomial -- a linear function).   While the analogous result for independent sets in $G(n,d/n)$ is also striking for the simplicity of an optimal algorithm (a random greedy algorithm~\cite{wormald1995differential,lauer2007large}), this algorithm still needs to explore arbitrarily large neighborhoods  to approach the density $\log d/d$ (and when implemented as a low-degree polynomial needs arbitrarily high but constant degree).      Because of this one could view the balanced independent set problem as a great test-case to try to `break' the local or low-degree hardness framework by  finding an efficient (non-local, not low-degree) algorithm that surpasses this barrier.  We discuss this more in Section~\ref{secFuture}.

\subsection{Local and low-degree algorithms}

Here we formally define both local and low-degree algorithms for finding independent sets in bipartite graphs before stating our main results precisely in Section~\ref{Subsec:Mainresults}.

\subsubsection{Local algorithms}
Roughly speaking, a local algorithm to identify a subset $S$ of vertices in a graph $G$ (say, an independent set) works as follows.  Each vertex of the graph is assigned an iid $\text{Uniform}[0,1]$ label; the decision of whether a  given vertex $v$ is in $S$ is a function of only the graph structure and random labels in the depth-$s$ neighborhood of $v$; such an algorithm is $s$-local.  The framework of local algorithms models parallel and distributed computation~\cite{luby1985simple,linial1992locality,parnas2007approximating,nguyen2008constant,hoppen2018local,bordenave2022detection}  and is closely related to the study of (finitary) factor of iid processes in probability theory~\cite{elek2010borel,hatami2014limits,lyons2017factors,brandt2022local}.

More formally, an $s$-\textit{local algorithm} for finding independent sets in a graph is defined by a measurable $s$-\textit{local function}  $g = g(H, \mathbf x)$ that takes as an input a rooted graph $H$ with root $r$ and depth at most $s$ and a vector $\mathbf x \in [0,1]^{V(H)}$ of labels, and outputs a bit $0$ (`out') or $1$ (`in'). 
We  apply such a local function to a graph $G$ by assigning the labels independently to each vertex from a Uniform$[0,1]$ distribution and then applying the local function to the depth-$s$ neighborhood of each vertex (with the appropriate restriction of the labels). 
To ensure that a local algorithm returns an independent set, we require that for any graph $G$, and any $\mathbf x \in [0,1]^{V(G)}$, if we apply the local function $g(\cdot,\cdot)$ in this way, the set of vertices on which the output is $1$ must form an independent set.  A local algorithm is an $s$-local algorithm for some constant $s$.

 A simple but instructive example of a $1$-local algorithm for independent sets is the function that returns $1$ if the label of the root is strictly smaller than the label of each of its neighbors.   If $G$ has maximum degree $d$, then this simple $1$-local algorithm will return an independent set of density at least $\frac{1}{d+1}$.    

The measure of performance of a local algorithm for independent sets in a random graph is the typical size or density of an independent set returned by the algorithm, with high probability over both the random graph and the random labels $\mathbf x$.    A striking fact is that local algorithms can find independent sets of density $(1+ o_d(1))\log d/d$ in $G(n,d/n)$ but cannot find independent sets of density $(1+\epsilon) \log d/d$ for any $\epsilon>0$ fixed and $d$ large enough~\cite{gamarnik2017limits,rahman2017local}, thus exhibiting  the factor $2$ statistical--computational gap conjectured to hold for \textit{all} efficient algorithms.

Here we consider local algorithms for bipartite graphs.  One could apply an $s$-local function $g$ as above to any bipartite graph, and to $\Gdb$ in particular.  However, such an algorithm would be blind to the bipartite structure; the random bipartite graph $\Gdb$ looks locally just like the random graph $G(n,d/n)$ (in the sense of Benjamini-Schramm~\cite{benjamini2001recurrence}), and thus such local algorithms cannot recover the trivial density $1/2$ independent set.

Instead,  we give our local algorithms knowledge of the global bipartition, and allow for two different local functions, $g_{\ell}, g_{r}$ to be applied to vertices in the partition sets $L,R$ respectively.  The pair defines a bipartite $s$-local algorithm if for any bipartite graph $G$ with bipartition $(L,R)$, and any $\mathbf x \in [0,1]^{V(G)}$, applying $g_{\ell}$ to all vertices in $L$ and $g_{r}$ to all vertices in $R$ returns an independent set.   We call such a pair of $s$-local functions \textit{compatible}.

We observe that such bipartite local algorithms can find a maximum size independent set in $\Gdb$.
\begin{observation}
There is a bipartite $0$-local algorithm that finds an independent set of density $1/2$ in $\Gdb$; this is achieved by the pair $g_{\ell} \equiv 0 ,g_{r} \equiv 1$ or vice versa.
\end{observation}

How can we find a \textit{balanced} independent set with a bipartite local algorithm?  Since the algorithms are local we cannot impose a strict global cardinality constraint.  Instead, given an independent set $I$ in a bipartite graph with partition $(L,R)$ we can always find a balanced independent set $I' \subseteq I$ by removing arbitrary vertices from $I$ on one side of the partition until it is balanced.  In this way, we can take the `balanced' density of $I$ to be $ \min \left \{  \frac{|I \cap L|}{n}, \frac{|I \cap R|}{n} \right \}$.  The balanced size of the trivial density $1/2$ independent sets is thus $0$.   
More generally, the $\gamma$-balanced density of an independent set $I$ is the density of the largest $\gamma$-balanced independent set $I' \subseteq I$.

Because sparse random graphs are locally treelike, the performance of local algorithms on the random graph $G(n,d/n)$ is determined, to first order, by the expectation of the corresponding local function $g(\cdot, \cdot)$ evaluated at the root of a $\text{Pois}(d)$ Galton--Watson tree (see~\cite{rahman2017local} for details).  
 
The same is true for bipartite local algorithms, and we make this precise in the following definition and lemma.

\begin{Definition}
\label{defGamBal}
    The $\gamma$-balanced value of a pair of compatible, $s$-local functions $g_\ell, g_r$ on the $\text{Pois}(d)$  Galton--Watson tree is
    \begin{equation}
        \alpha_{d,\gamma}(g_\ell, g_r):= \frac{1}{2} \min \left\{  \frac{1}{\gamma} \E [g_\ell(\bT,\mathbf{x})] , \frac{1}{1-\gamma}  \E [g_r(\bT,\mathbf{x})]   \right \}  \,,
    \end{equation}
    where the expectation is over a rooted $\text{Pois}(d)$ Galton--Watson tree $\bT$ and iid $\text{Unif}[0,1]$ labels $\mathbf x$ on its vertices.
\end{Definition}
The minimum in the formula is to take into account which side one should remove vertices from to make the independent set $\gamma$-balanced.
\begin{lemma}
\label{Lemma:GamValueRandom}
    Suppose $g_{\ell},g_{r}$ is a pair of compatible $s$-local functions for independent sets.  Let $X$ be the random independent set obtained by applying $g_{\ell}, g_r$ to $\Gdb$, and let $X_\gamma$ be the largest $\gamma$-balanced independent set so that $X_{\gamma}\subseteq X$.  Then whp over the randomness in the graph and in the random labels $\mathbf x$, 
    \[\frac{|X_{\gamma}|}{2n} =  \alpha_{d,\gamma}(g_{\ell},g_{r})  + o(1) \,.  \]
\end{lemma}
Thus determining the performance, up to first order, of local algorithms for finding $\gamma$-balanced independent sets in $\Gdb$ is reduced to determining the largest possible $\alpha_{d,\gamma}$ over compatible pairs of $s$-local functions for $s$ arbitrarily large.

\subsubsection{Low-degree algorithms}
Low-degree algorithms are those than can be implemented as a vector of low-degree, multivariate polynomials of the input (in our case the entries of the adjacency matrix of the random graph).  While at first glance this might not seem to be the most natural class of algorithms, it turns out that they capture the best known performance of efficient algorithms for many problems with statistical--computational gaps, and many important algorithmic techniques can be implemented in a low-degree fashion.
See~\cite{hopkins2018statistical} for a survey of the low-degree framework and~ e.g.~\cite{hopkins2017efficient,hopkins2017power,gamarnik2020low,wein2022optimal,bresler2022algorithmic,schramm2022computational,ding2023subexponential} for some results about the power of low-degree algorithms for problems exhibiting statistical--computational gaps.

We now define these algorithms.   We say a function $f:\mathbb{R}^m \to \mathbb{R}^{2n}$ is a \textit{polynomial of degree (at most) D} if it can be written as $f(A) = (f_{1}(A), \cdots, f_{2n}(A))$, where each $f_i:\mathbb{R}^m \to \mathbb{R}$ is a multivariate polynomial of degree at most $D$. For some probability space $(\Omega, P_{\omega})$, we say $f:\mathbb{R}^m \times \Omega \to \mathbb{R}^{2n}$ is a \textit{random polynomial} if $f(\cdot, \omega)$ is a degree $D$ polynomial for each $\omega \in \Omega$, i.e., the coefficients of $f$ are random but do not depend on $A$.
For our purposes, we let the input of $f$ be an $2n$-vertex balanced bipartite graph $G$ encoded as $A \in \{0,1\}^{m}$ with $m = n^2$ and each entry of $A$  the indicator of a particular edge of $G$. 

What does it mean for such a polynomial function to find an independent set in a graph?  We follow, with some modifications for the bipartite setting, the definitions and notions used in~\cite{wein2022optimal}, in which   a rounding procedure is used to produce an independent set of $G$ based on $f(A)$.

\begin{Definition} \label{Def:V_eta_f}
    Let $f:\mathbb{R}^m \to \mathbb{R}^{2n}$ be a random polynomial, with $m = n^2$. For $A \in \{0,1\}^m$ indicating the edges of a bipartite graph $G$ on $n+n$ vertices and $\eta >0$, let $V^{\eta}_f(A,\omega)$ be the independent set in the graph $G$ constructed as follows. Let 
    \begin{align*}
        I &=\big\{ i \in [2n]: f_i(A,\omega) \geq 1\big\} \\
        \Tilde{I} &= \big\{i \in I:\text{i has no neighbors in I in  $G$}\big\} \,  ,  \, \, \text{ and} \\
         J &= \left\{ i \in [2n]: f_i(A,\omega) \in \left(\frac{1}{2},1\right)\right\} \,.
    \end{align*}
    Then define
    \begin{align*}
        V^{\eta}_f(A,\omega) =
        \begin{cases}
            \Tilde{I} & \text{if } |I \setminus  \Tilde{I}|+|J| \leq \eta n; \\
            \varnothing & \text{otherwise}.
        \end{cases}
    \end{align*}
\end{Definition}
Intuitively speaking, a vertex $i$ is in the independent set if the output of $f_i$ $\geq 1$ and is not in the independent set if the output of $f_i$  is at most $\frac{1}{2}$. Some errors are allowed, as long as there are no more than $\eta n$ many vertices which have value in $(1/2,1)$ or violate the independent set constraint. Based this mapping of the polynomial's output to independent sets, we define what is means for a random polynomial to find a balanced independent set of certain size.
\begin{Definition} \label{Def:f_Optimize_ind}
    For parameters $k_{\ell}, k_r > 0$, $\delta \in [0,1]$, $\xi \geq 1$, a random function $f : \mathbb{R}^m \to \mathbb{R}^{2n}$ is said to $(k_{\ell}, k_{r}, \delta, \xi, \eta)$-optimize the independent set problem in $\Gdb$ if the following are satisfied when $A \sim \Gdb$:
    \begin{itemize}
        \item $\mathbb{E}_{A,\omega}\left[\|f(A,\omega)\|^2\right] \leq \xi(k_{\ell}+ k_r)$ and
        \item $\mathbb{P}_{A,\omega}\left[|V^\eta_f(A,\omega) \cap L| \geq k_{\ell} \text{ and }|V^\eta_f(A,\omega) \cap R| \geq k_r\right] \geq 1-\delta$
    \end{itemize}
\end{Definition}
Here $k_{\ell}, k_r$ are the desired sizes of the intersections of the independent set on each side of the bipartition, $\delta$ is the algorithm's failure probability, $\xi$ is a normalization parameter, and $\eta$ is the error tolerance parameter of the rounding procedure. 

The trivial algorithm to find a density $1/2$ independent set in $\Gdb$ can be implemented as a (deterministic) degree-$0$ algorithm in the following way: $f_i(A) = 1$ if $i \in L$ and $0$ if $i \in R$; that is, $f$ is  a constant vector.  This function $(n,0,0,1,0)$-optimizes the independent set problem in $\Gdb$.

\subsection{Main results}
\label{Subsec:Mainresults}
    
    We first establish an existential result on  the typical density of the largest $\gamma$-balanced independent set in $\Gdb$.

    \begin{Theorem}\label{Thm:balexist}
     Let $\Xg$ be the size of the largest $\gamma$-balanced independent set in $\Gdb$.
        For every $\epsilon>0$ and sufficiently large $d$, with probability $1-o(1)$ as  $n \to \infty$ 
        $$ \left( \frac{1}{2 \gamma (1-\gamma)} -\eps \right) \frac{\log d}{d} \le \frac{\Xg}{2n} \le  \left( \frac{1}{2 \gamma (1-\gamma)} +\eps \right)\frac{\log d}{d}.$$
        In particular, the density of the largest balanced independent set in $\Gdb$ is $(2+ o_d(1)) \frac{\log d}{d}$.
    \end{Theorem}
  We prove Theorem~\ref{Thm:balexist} in Section~\ref{Sec:Exist} using the first-moment method for the upper bound and adapting Freize's argument from~\cite{frieze1990independence} for the lower bound.

Now let 
\[ \hat {\alpha}_{d, \gamma}:= \sup \left\{ \alpha_{d,\gamma} (g_{\ell}, g_{r}): s \ge 0, g_{\ell}, g_{r} \text{ compatible bipartite s-local functions}    \right \} \,. \] 
By Lemma~\ref{Lemma:GamValueRandom} the value $\hat {\alpha}_{d, \gamma}$ is the optimal density (to first order) of a $\gamma$-balanced independent set that a bipartite local algorithm can find whp.

Our next result is that local algorithms are only half optimal for balanced independent sets in random bipartite graphs and $(1-\gamma)$-optimal for $\gamma$-balanced independent sets. 
\begin{Theorem} \label{Thm:balhalfopt}
    For any $\gamma \in (0,1/2]$,
    $$\hat{\alpha}_{d,\gamma} = \left(\frac{1}{2 \gamma}+o_d(1) \right)\frac{\log d}{d}  \,.$$
    Moreover, the lower bound is achieved by a bipartite $1$-local algorithm.
\end{Theorem}
In particular Theorem~\ref{Thm:balhalfopt} says that $\hat{\alpha}_{d,1/2} = (1+ o_d(1)) \log d/ d$.  Note that a $1$-local algorithm is very simple: whether a given vertex belongs to the independent set depends only on the label of the vertex along with the number and labels of its neighbors.

\bigskip

Next we extend the upper and lower bounds to the class of low-degree algorithms.

\begin{Theorem} \label{Thm:imposs_achieve_lowdeg}
    For any  $\varepsilon >0$ the following hold:
    \begin{itemize}
        \item For any $d>0$, there exists $\xi \ge 1$  so that if  $k_{\ell} \leq (1-\varepsilon)\frac{\log d}{d} n$ and $k_r \leq (1-\eps) d^{\varepsilon-1}n$, there is a degree-$1$ polynomial (a linear function) that $(k_{\ell},k_r, o_n(1), \xi, 0)$-optimizes the independent set problem in $\Gdb$. 

        \item  If $\min \{k_{\ell},k_r\} \geq (1+\varepsilon)\frac{\log d}{d} n$, then there exists $d_0>0$ such that for any $d \geq d_0$ there exists $n_0 >0$, $\eta >0$ $C_1>0$ and $C_2>0$ such that for any $n \geq n_0$, $\xi \geq 1$, $1 \leq D \leq \frac{C_1 n}{\xi \log n}$ and $\delta \leq \exp(-C_2 \xi D \log n)$, there is no random polynomial of degree $D$  that $(k_{\ell}, k_r, \delta, \xi, \eta\big)$-optimizes the independent set problem in $\Gdb$. 
    \end{itemize}
\end{Theorem}

We remark that the first part of Theorem \ref{Thm:imposs_achieve_lowdeg} also holds when $k_r \leq (1-\varepsilon)\frac{\log d}{d} n$ and $k_{\ell} \le (1-\eps) d^{\varepsilon-1}n$. The result says that as long as the density on one side is a constant factor less that $\frac{\log d}{d} n$, then the density on the other side can be polynomially larger than $\log d/d$.

As mentioned above, the hardness results of Theorems~\ref{Thm:balhalfopt} and~\ref{Thm:imposs_achieve_lowdeg} are proved used the Overlap Gap Property (OGP).
 Seeking to provide rigorous evidence of the statistical--computational gap in finding large independent sets in sparse random graphs (and addressing the conjecture in~\cite{hatami2014limits}), Gamarnik and Sudan~\cite{gamarnik2017limits} developed the OGP concept  and proved a rigorous link to the failure of  local algorithms for search problems.  Gamarnik and Sudan used this framework to show that local algorithms fail to find independent sets of size $(\sqrt{2} +\eps) \frac{\log d}{d}$ for any $\eps>0$ fixed and $d$ large~\cite{gamarnik2017limits}.  Rahman and Virag then ruled out local algorithms for finding find independent sets of size $(1 +\eps) \frac{\log d}{d}$ using a more sophisticated concept of `ensemble-$m$-OGP'.  This result is optimal in the leading constant as there exist local algorithms that find independent sets of size $(1 -\eps) \frac{\log d}{d}$ in sparse random graphs~\cite{rahman2017local}.   The resulting state of knowledge is fairly satisfying -- local algorithms match the performance of the best known efficient algorithms for this problem, and  local algorithms provably cannot do better.  This provides some evidence that this problem may be computationally intractable beyond this point, and also serves as a challenge to search for algorithmic techniques that cannot be captured by the local framework.  Even stronger evidence was provided by Wein~\cite{wein2022optimal}, following~\cite{gamarnik2021overlap,gamarnik2020low}, who showed that low-degree algorithms for independent sets in $G(n,d/n)$ are also only half optimal.  See~\cite{gamarnik2021overlapT} for a survey of the OGP framework.

 Our main contribution is in identifying a problem with a statistical--computational gap in which the best known algorithm is so trivial, both in its description and in its implementation as a $1$-local algorithm and a linear function.  We also demonstrate, in a new setting, the phenomenon of cardinality constraints inducing complexity, or of hard problems lurking inside easy ones~\cite{mezard1987mean,feige2004hardness,carlson2022computational}.    Once the correct definitions are chosen the proofs are streamlined, with the hardness results following~\cite{wein2022optimal}.  One additional contribution is in  the very explicit constructions of the local and low-degree algorithms achieving the lower bounds; this may be useful pedagogically.

\subsection{Future Directions}
\label{secFuture}

A bit of intuition that inspired the current results is that the structure of balanced independent sets in $\Gdb$ is  close to that of independent sets in $G(n, d/n)$.  On the level of first- and second-moment calculations, this is indeed the case, as we will see in the proof of Theorem~\ref{Thm:balexist} and the OGP analysis below.   Similarly, these models are very similar on the level of the cavity method from statistical physics, e.g.~\cite{barbier2013hard}.  However, establishing some basic fundamentals of balanced independent sets seems to be challenging even in comparison to independent sets in $G(n,d/n)$. One such fundamental question is whether for fixed $d$, the balanced independence ratio of sparse random bipartite graphs has a limit.

    \begin{Conjecture}\label{Conj:converge}
        For each $d>0$, there exists $\alpha_{d,1/2} >0$ so that in $\Gdb$,
       $ \frac{X_{1/2}}{n}$  converges in probability to $\alpha_{d,bal}$ as $n\to \infty$.
    \end{Conjecture}
For independent sets in $G(n,d/n)$ the existence of such a limit was proved by Bayati, Gamarnik, and Tetali~\cite{bayati2010combinatorial} using a beautiful `combinatorial interpolation' technique based on~\cite{guerra2002thermodynamic,franz2003replica}.  However, an attempt to adapt this proof to the balanced bipartite case fails: the crucial inequality is not always valid.  This leaves Conjecture~\ref{Conj:converge} open and also reveals that one of the crucial methods we have for studying optimization problems on sparse random graphs is rather brittle to the specifications of the model.  A similar phenomenon occurs for the min bisection problem in sparse random graphs, for which there are conjectures of the value limiting density but the even the existence of the limit remains open~\cite{zdeborova2010conjecture,dembo2017extremal}.

Perhaps the biggest open question raised by this paper is whether any efficient algorithm can surpass the $\log d /d$ density barrier for balanced independent sets in $\Gdb$.  For independent sets in $G(n,d/n)$ this is of course a long-standing challenge, but it could well be that the balanced independent set problem is easier -- this is suggested by the complete simplicity of the algorithm that achieves density $\log d/d$.

\section{Local  and low-degree algorithms on bipartite graphs}
\label{Sec:Prelim}

In this section we prove the positive results for local and low-degree algorithms.

\subsection{A simple 1-local algorithm} \label{Sec:Achievability_local}
Here we interpret and analyze the simple  algorithm from~\cite{chakraborti2021extremal} as a $1$-local bipartite algorithm to prove the positive algorithmic results.
\begin{Proposition}
    For all $p \in [0,1]$, $d >0$,  there is a pair of $1$-local, compatible functions $g_{\ell}, g_r$ with  $\E[ g_\ell(\bT, \mathbf{x})] = p$ and $\E [g_r(\bT,\mathbf x)] = \exp( - dp)$, where $\bT$ is a $\text{Pois}(d)$ Galton--Watson tree and $\mathbf{x}$ are iid $\text{Unif}[0,1]$ labels on its vertices.  
\end{Proposition}
\begin{proof}
The algorithm has one parameter, $p \in [0,1]$, and proceeds as follows:
\begin{enumerate}
    \item Include each vertex from $L$ in the independent set independently with probability $p$.
    \item Add to the independent set each vertex from $R$ with no neighbors currently in the independent set. 
\end{enumerate}

We can implement this algorithm as a $1$-local algorithm by taking the following $1$-local functions:
\begin{align*}
    g_{\ell}(H_v,\mathbf{x}) = 
    \begin{cases}
      1 & \text{if } \mathbf{x}_v \leq p;\\
      0 & \text{otherwise,}
    \end{cases} \text{   and  }~~~~
    g_{r}(H_v,\mathbf{x}) = 
    \begin{cases}
      1 & \text{if } \mathbf{x}_u > p , \forall u \in N_{H_v}(v); \\
      0 & \text{otherwise,}
    \end{cases}    
\end{align*}
where $H_v$ is a depth-$1$ graph rooted at $v$.

 The functions $g_{\ell}, g_r$ are compatible  by definition.   Calculating their expectations on the $\text{Pois}(d)$ Galton--Watson tree is also immediate:
 $\E[g_\ell(\bT,\mathbf{x})] = p$ and $\E[g_r(\bT,\mathbf{x})]= e^{-p d}$, since the number of neighbors of the root with label at most $p$ has a $\text{Pois}(dp)$ distribution.
\end{proof}

From this we can immediately deduce the lower bound on $ \hat \alpha_{d,\gamma}$ in Theorem~\ref{Thm:balhalfopt}.
\begin{Corollary}
\label{CorlocalLB}
For any $\gamma \in (0,1/2]$,
\[ \hat \alpha_{d,\gamma} \ge \left( \frac{1}{2\gamma} +o_d(1) \right) \frac{\log d}{d} \,. \]
\end{Corollary}
\begin{proof}
  Fix $\eps>0$ and let $p= (1-\eps) \frac{\log d}{d}$.  Then   $\E[g_\ell(\bT,\mathbf{x})] = (1-\eps) \frac{\log d}{d}$ and $\E[g_r(\bT,\mathbf{x})]= d^{\eps -1}$, where $\bT$ is the $\text{Pois}(d)$ Galton--Watson tree.  Applying Definition~\ref{defGamBal}, we see that, for $\eps$ small enough and $d$ large, 
  \[ \alpha_{d,\gamma} (g_\ell,g_r) = \frac{1}{2 \gamma} (1-\eps)\frac{\log d}{d}= \left( \frac{1}{2\gamma} +o_d(1) \right) \frac{\log d}{d} \,. \qedhere\]
\end{proof}

Finally we 
conclude this section by proving Lemma~\ref{Lemma:GamValueRandom}, establishing that $\alpha_{d,\gamma}(g_\ell,g_r)$ determines to first order the performance of a bipartite local algorithm for finding $\gamma$-balanced independent sets in $\Gdb$.

\begin{proof}[Proof of Lemma~\ref{Lemma:GamValueRandom}]

We first show that expectation of $|X_\gamma|$ is $2n\alpha_{d,\gamma}(g_{\ell}, g_r)+o(n)$. Assume that $ \frac{1}{\gamma} \E [g_{\ell}(\bT,\mathbf{x})] \geq \frac{1}{1-\gamma} \E [g_{r}(\bT,\mathbf{x})]$, then we have $\alpha_{d,\gamma}(g_{\ell}, g_r) = \frac{1}{2(1-\gamma)}\E [g_{r}(\bT,\mathbf{x})]$. It also means that we need to delete $|X \cap L|-\frac{\gamma}{1-\gamma}|X \cap R|$ many vertices in $X \cap L$ to obtain $X_{\gamma}$, so that $\frac{|X_{\gamma} \cap L|}{|X_{\gamma} \cap R|} = \frac{\gamma}{1-\gamma}$ as desired. It follows that
\begin{align*}
    \E\big[|X_\gamma|\big] &= \E\left[|X| - \left(|X \cap L|-\frac{\gamma}{1-\gamma}|X \cap R|\right)\right] = \frac{1}{1-\gamma}\E\left[|X \cap R|\right]\\&= \frac{1}{1-\gamma}\E [g_{r}(\bT,\mathbf{x})]+o(n) = 2n\alpha_{d,\gamma}(g_{\ell}, g_r)+o(n).
\end{align*}
The case of  $ \frac{1}{\gamma} \E [g_{\ell}(\bT,\mathbf{x})] \leq \frac{1}{1-\gamma} \E [g_{r}(\bT,\mathbf{x})]$ is similar.

Next, we show the concentration by showing that $\text{Var}\left[|X \cap R|\right]$ and $\text{Var}\left[|X \cap R|\right]$ are both $ O(n)$. Let $G \sim \Gdb$, for $v_i \in R$, let $E_{v_i}$ be the set of edges incident to $v_i$. Then there is a measurable function $F$ such that 
\[
    \sum_{v \in R}g_{r}(H_v,\mathbf{x}) = F(E_{v_1}, E_{v_2}, \cdots, E_{v_n}),
\]
where $H_v$ is a depth-$s$ graph rooted at $v$ in $G$.
Let $W = F(E_{v_1}, E_{v_2}, \cdots, E_{v_n})$ denote this sum, and for each $v_i \in R$ let $W_{v_i}$ be the same sum but $E_{v_i}$ is empty,
\[
    W_{v_i} = F(E_{v_1}, \cdots, E_{v_{i-1}}, \varnothing, E_{v_{i+1}}, \cdots, E_{v_n}).
\]
Equivalently, if $G_{v_i}$ indicates the graph $G$ where all edges incidents to $v_i$ have been deleted, and where $H_v^{v_i}$ is a depth-$s$ graph rooted at $v$ in $G_{v_i}$.we have
\[
    W_{v_i} = \sum_{v \in R}g_{r}(H_v^{v_i},\mathbf{x}).
\]
Since $(E_{v_i})_{v_i \in R}$ are independent and $0 \le g_{r}(G,v) \leq 1$ for all $v$ in $R$, we can apply the moment inequality \cite[Theorem 15.5]{boucheron2013concentration} in this setting: there exists some constant 
$c_1$ such that for
\[
    \mathbb{E}\left[\Bigg|\sum_{v \in R} g_{r}(H_v,\mathbf{x}) - \mathbb{E} \sum_{v \in R} g_{r}(H_v,\mathbf{x})\Bigg|^2\right] \leq c_1\mathbb{E}\left[\sum_{v_i \in R} (W-W_{v_i})^2\right]
\]

Now fix $v$ and $v_i$ in $R$, $g_{r}$ is $s$-local, hence $g_{r}(H_v,\mathbf{x}) - g_{r}(H_v^{v_i},\mathbf{x})=0$, except possibly $v$ is in the $s$-neighbor, denoted as $N_{s}(G,v_i)$,  of $v_i$. This implies 
\begin{align*}
    |W-W_{v_i}| \leq \sum_{v\in N_{s}(G,v_i)} g_{r}(H_v,\mathbf{x}) + g_{r}(H_v^{v_i},\mathbf{x}) \leq 2\left|N_{s}(G,v_i)\right|.
\end{align*}
It follows that
\begin{align*}
    \mathbb{E}\left[\sum_{v_i \in R} (W-W_{v_i})^2\right] \leq 4  \mathbb{E}\left[\sum_{v_i \in R}\left|N_{s}(G,v_i)\right|^2\right] \leq 4n\sqrt{\E\left[\left|N_{s}(G,v_i)\right|^4\right]}, 
\end{align*}
where Cauchy-Schwarz inequality is used in the second inequality. We can bound the moments of the size of the first $s$ generations in a Galton-Watson tree (which dominates the size of the $s$-neighborhood of a vertex in $\Gdb$) from above \cite[(11.8)]{bordenave2022detection}:  there exists some constant $c_2>0$ such that 
\[
    \E\left[\left|N_{s}(G,v_i)\right|^4\right] \leq c_2 d^{4s}.
\]
Putting all together,
\begin{align*}
    \mathbb{E}\left[\Bigg|\sum_{v \in R} g_{r}(H_v,\mathbf{x}) - \mathbb{E} \sum_{v \in R} g_{r}(H_v,\mathbf{x})\Bigg|^2\right] \leq 4c_1\sqrt{c_2}d^{2s} n.
\end{align*}
This means $\text{Var}\left[|X \cap R|\right] = O(n)$; similarly $\text{Var}\left[|X \cap L|\right] = O(n)$, and hence the overall size concentration follows.
\end{proof}

\subsection{A simple degree-$1$ algorithm}
\label{Sec:achievability_lowdeg}
In this section, we construct a degree-1 polynomial, that is, a linear function, to prove the first part of Theorem \ref{Thm:imposs_achieve_lowdeg}. Namely, we will construct a degree-$1$ polynomial $F:\{0,1\}^{n^2} \to \R ^{2n}$ that $((1-\varepsilon)\frac{\log d}{d}n, (1-\eps)  d^{\varepsilon-1}n, \delta, \xi, 0)$-optimizes the independent set problem in $G$, for $\xi$ some constant that depends only on $d$, and $\del \to 0$ as $n \to \infty$. 

Without loss of generality, for a fixed $\varepsilon>0$, we let $k_{\ell} = (1-\varepsilon)\frac{\log d}{d}n$ and $k_r = (1-\eps) d^{\varepsilon-1}n$ (if $k_{\ell}, k_r$ are smaller we can zero out some coordinates to achieve the guarantees). Let $G \sim \Gdb$ be encoded as $A \in \{0,1\}^{n^2}$.

For the left partition, choose any subset $L_1 \subseteq L$ of size $k_{\ell}$ and let $F_{v}(A) = 1$ for all $v \in L_1$. Let $L_0 = L \setminus L_1$ and $F_{v}(A) = 0$ and all $v \in L_0$. Notice that $\forall v \in L$, $F_v$ has degree 0. By definition, $\sum_{v \in L} F_v(A) =  k_{\ell}$, for every input $A$.

For $u \in R$, define $F_u(A) := 1- \sum_{v \in L_1}A_{uv}$, which has degree 1. In the case of $u$ has no neighbors in $L_1$, $F_u$ matches  $g_{r}$, giving the value $1$. 

Now we check that $F$ achieves the conditions of Definition \ref{Def:f_Optimize_ind} for chosen parameters.  If $F_u$ is not $1$ it is at most $0$, and so there are no errors from values in $ (1/2,1)$.  By construction there are also no vertices that violate the independent set constraint.

With probability $1$, $|V^{\eta}_{F}(A) \cap L| = k_{\ell}$, and with probability $1- \exp(-\Theta(n))$, $|V^{\eta}_{F}(A) \cap R| \ge k_r $, since the vertices in $R$ are each included independently with probability $( 1- d/n)^{k_\ell}$, and so a binomial tail bound suffices.

For the expectation of the squared $2$-norm of $F$, we calculate     
\begin{align*}
 \mathbb{E}_{A}\left[F^2_u(A)\right]
 &= \sum_{N = 1-|L_1|}^1 N^2 \cdot\mathbb{P}[F_u(A)=N] \\
 &=  \sum_{N = 1-|L_1|}^1 N^2\cdot\mathbb{P}[d(u,L_1)=1-N] \quad \text{ where } d(u,L_1)\text{ is the number of neighbors of } u \text{ in } L_1\\
 &= \mathbb{P}[d(u,L_1)=0]+\sum_{N=2}^{|L_1|}(1-N)^2\cdot \mathbb{P}[d(u,L_1)=N] \\
 &\leq  d^{\varepsilon-1} + \sum_{N=2}^{\infty}(N-1)^2\cdot \mathbb{P}[d(u,L_1)=N]\\
 &=  d^{\varepsilon-1}+ O(d^2
 ) \, ,
 \end{align*}
 and so,
 \begin{align*}
 \mathbb{E}_{A}\big[\|F(A)\|^2\big] &= \mathbb{E}_{A}\left[\sum_{v \in G}F^2_v(A)\right] \\
    &=\sum_{v \in L} \mathbb{E}_{A}\left[F^2_v(A)\right]+ \sum_{v \in R} \mathbb{E}_{A}\left[F^2_v(A)\right] \\
    &\leq k_{\ell}+ O(n d^2)\\
   & \leq \xi(k_{\ell}+ k_r), \text{ for some $\xi = \xi(d) \geq 1$ large enough}.
\end{align*}
 
Thus $F $ $(k_{\ell},k_r, o_n(1), \xi, 0)$-optimizes the independent set problem in $\Gdb$.

\section{Lower bounds against local and low-degree algorithms}
\label{secLB}

To prove the failure of local and low-degree algorithms, we follow the strategy of~\cite{wein2022optimal}, with appropriate modifications of definitions and lemmas to account for the bipartite setting.
 In Section~\ref{Subsec:Forb_struc} we will show that a certain  forbidden structure  exists in correlated copies of $ \Gdb$ with only exponentially small probability. Then in Section~\ref{Subsec:D_stable}, we  define $D$-stable algorithms and show that low-degree algorithms and bipartite local algorithms are $D$-stable for appropriate choices of parameters. In the Section~\ref{secMainProof}, we will show that if some $D$-stable algorithm can find independent sets of density $(1+\varepsilon) \frac{\log d}{d}$ in both $L$ and $R$ with small enough error probability , then it results in a contradiction to the small probability of existence of the forbidden structure. Thus no $D$-stable algorithms can achieve $(1+\varepsilon) \frac{\log d}{d}$ density on both sides of $G$ with small error probability.

We first prove a lemma about the correlation of the size distribution of independent sets in $\Gdb$.
\begin{lemma} \label{Lemma:balancex+y>xy}
   For any independent set $I$ in $\Gdb$, define $\alpha_L$ and $\alpha_R$ as $|I \cap L| = \alpha_L \frac{\log d}{d} n$, $|I \cap R| = \alpha_R \frac{\log d}{d} n$. If $d$ is sufficiently large, then with high probability over $\Gdb$ $\alpha_L+\alpha_R \geq \alpha_L\alpha_R$ for all independent sets $I$.
\end{lemma}

\begin{proof}
 For convenience we denote $\varphi = \frac{\log d}{ d}$.
    Let $X$ be the number of independent set in $\Gdb$ such that $ \alpha_L+\alpha_R < \alpha_L\alpha_R$. We want to show that for large enough $d$,
    \[
        \mathbb{P}[X = 0] = 1- o(1) \text{ as } n \to \infty
    \]
    We use Markov's inequality to bound
    \[
         \mathbb{P}[X > 0] \leq  \mathbb{E}[X] = {n \choose \alpha_L \varphi n} {n \choose \alpha_R \varphi n} \left(1-\frac{d}{n}\right)^{\alpha_L\alpha_R\varphi^2n^2}.
    \]
    Now we want to show that $\lim_{n \to \infty} \mathbb{P}[X > 0] = 0$, and to do this it suffices to show that $\limsup_{n \to \infty} n^{-1} \log( \mathbb{E}[X]) \leq -\delta$ for some $\delta > 0$. Using Stirling's approximation and the fact that $\log(1-x) \approx -x$ as $x \to 0$, we have
\begin{align*}
    &~~~~\limsup_{n \to \infty} \frac{1}{n}\log\big( \mathbb{E}[X_f]\big) \\
    &= -\alpha_L \varphi\log( \alpha_L \varphi ) -(1-\alpha_L \varphi)\log( 1-\alpha_L \varphi )-\alpha_R \varphi\log( \alpha_R \varphi ) -(1-\alpha_R \varphi)\log( 1-\alpha_R \varphi ) - \alpha_L\alpha_R\varphi^2 \\
    &= -\alpha_L \varphi\log( \alpha_L \varphi ) + \alpha_L\varphi(1-\alpha_L \varphi)-\alpha_R \varphi\log( \alpha_R \varphi ) +\alpha_R \varphi(1-\alpha_R \varphi) - \alpha_L\alpha_R\varphi^2
\end{align*}
Note that the asymptotic leading term is $(\alpha_L+\alpha_R-\alpha_L\alpha_R) \frac{\log^{2}d}{d}$ which is negative by assumption and hence the desired negative limit follows.
\end{proof}

\subsection{The Forbidden Structure}
\label{Subsec:Forb_struc}
Here we define a sequence of correlated random graphs and a corresponding forbidden structures. Recall we represent a $2n$-vertex balanced bipartite graph $G = (L,R)$ by $A \in \{0,1\}^m$, where $m = n^2$ and $A_1, \cdots, A_m$ are the indicators of edges of $\Gdb$ in some fixed order. 

\begin{Definition} \label{Def:interpolationpath}
    For $T \in \N$, consider a length $T$ interpolation path $A^{(0)}, \cdots, A^{(T)}$ constructed as follows:  $A^{(0)} \sim \Gdb$; then for each $1 \leq t \leq T$, $A^{(t)}$ is obtained by resampling the coordinate $\sigma(t) \in [m]$ from Ber(d/n). Here $\sigma(t) = t-k_t m$ where $k_t$ is the unique integer making $1 \leq \sigma(t) \leq m$. 
\end{Definition}
It easy to see that the marginal distribution of each $A^{(t)}$ is $ \Gdb$ (identifiying the graph with the indicator vector of its edges). 
Moreover,  in any $m$ consecutive steps all possible  edges are resampled, which implies that $A^{(t+m)}$ is independent from $A^{(0)}, \cdots, A^{(t)}$. 

Now we prove  that a certain structure  of independent sets across the correlated sequence of random graphs exists only with exponentially small probability (compare to~\cite[Proposition 2.3]{wein2022optimal}).
\begin{Proposition} \label{Prop:forbidden}
    Fix  $\varepsilon>0$ and $K \in \N$ with $K \geq 9/\varepsilon^2+1$. Consider the interpolation path $A^{(0)}, \cdots, A^{(T)}$ of any length $T = n^{O(1)}$. If $d = d(\varepsilon, K)>0$ is sufficiently large, then with probability $1-\exp\big(-\Omega(n)\big)$ there does not exist a sequence of subsets $S_1, \cdots, S_K$ of $L \cup R$ satisfying the following:
    \begin{enumerate}
        \item for each $k \in [K]$ there exists $0 \leq t_k \leq T$ such that $S_k$ is an independent set in $A^{(t_k)}$.
        \item $|S_k \cap L| \geq (1+\varepsilon)\frac{\log d}{d}n$ and $|S_k \cap R| \geq (1+\varepsilon)\frac{\log d}{d}n$, $\forall k 
        \in [K]$.
        \item $\big|S_k \setminus \cup_{i<k}S_i\big| \in \left[\frac{\varepsilon}{4}\frac{\log d}{d}n, \frac{\varepsilon}{2}\frac{\log d}{d}n\right]$ for all $2 \leq k \leq K$.
    \end{enumerate}

\end{Proposition}

\begin{proof}
    Let $N$ denote the number of tuples  $(S_1, \cdots, S_K)$ satisfying the conditions $(1)$-$(3)$. We will compute $\E[N]$ and show that it is exponentially small. Let $\phi = \frac{\log d}{d}n$. Let $a_{L,k}, a_{R,k}$ and $b_k$ be defined by $a_{L,k}\phi = |S_k \cap L|$, $a_{R,k}\phi = |S_k \cap R|$ and $b_k \phi = \big|S_k \setminus \cup_{i<k}S_i\big|$. $(2)$ and $(3)$ state that $a_{L,k} \geq 1+\varepsilon$, $a_{R,k} \geq 1+\varepsilon$ and $b_k \in [\frac{\varepsilon}{4},\frac{\varepsilon}{2}]$. Also let $c$ be defined by $c\phi = |\cup_{k}S_k|$, then (3) implies that $c \leq (a_{L,1}+a_{R,1})/2+(K-1)\varepsilon$. By Lemma \ref{Lemma:balancex+y>xy} we can assume that $a_{L,k},a_{R,k} \leq 1+\frac{1}{\varepsilon}$. Thus, $c=c(K,\varepsilon)$ is a bounded constant independent of $d$. To compute $N$ we still need to specify the size of $\big(S_k \setminus \cup_{i<k}S_i\big)\cap L$, for which we define $b_{L,k}$ by $b_{L,k} \phi = \left|\big(S_k \setminus \cup_{i<k}S_i\big)\cap L\right|$, and note that $b_{L,k} \in [0,b_k]$. There are at most $(n/\epsilon)^{2K}\cdot n^{2K}$ choices for the values $\{a_{L,k}\}$, $\{a_{R,k}\}$, $\{b_{k}\}$ and $\{b_{L,k}\}$. Once they are fixed, the number of ways to choose $\{S_k\}$ is at most
    \begin{align*}
        &~~~~{n \choose a_{L,1}\phi}{n \choose a_{R,1}\phi} \prod_{k=2}^K{n \choose b_{L,k}\phi}{n \choose (b_k-b_{L,k})\phi} {c\phi \choose (a_{L,k}-b_{L,k})\phi}{c\phi \choose (a_{R,k}-b_k+b_{L,k})\phi} \\[1em]
        & \leq \left(\frac{en}{a_{L,1}\phi}\right)^{a_{L,1}\phi}\left(\frac{en}{a_{R,1}\phi}\right)^{a_{R,1}\phi} \prod_{k=2}^K
        \left(\frac{en}{b_{L,k}\phi}\right)^{b_{L,k}\phi}
        \left(\frac{en}{(b_k-b_{L,k})\phi}\right)^{(b_k-b_{L,k})\phi} \times\\
        &~~~\left(\frac{ec}{a_{L,k}-b_{L,k}}\right)^{(a_{L,k}-b_{L,k})\phi}
        \left(\frac{ec}{a_{R,k}-b_k+b_{L,k}}\right)^{(a_{R,k}-b_k+b_{L,k})\phi} \\[1em]
        & = \exp\left\{a_{L,1}\phi \log \left(\frac{ed}{a_{L,1}\log d}\right)+a_{R,1}\phi \log \left(\frac{ed}{a_{R,1}\log d}\right)\right\} \times\\
        &~~~\exp \left\{\phi\sum_{k=2}^K\left[b_{L,k}\log\left(\frac{ed}{b_{L,k}\log d}\right) + (b_k-b_{L,k}) \log\left(\frac{ed}{(b_k-b_{L,k})\log d}\right)\right]\right\} \times \\
        &~~~\exp \left\{\phi\sum_{k=2}^K\left[(a_{L,k}-b_{L,k})\log\left(\frac{ec}{a_{L,k}-b_{L,k}}\right) + (a_{R,k}-b_k+b_{L,k})\log\left(\frac{ec}{a_{R,k}-b_k+b_{L,k}}\right)\right]\right\} \\
        &= \exp\left\{\phi \log d\left(a_{L,1}+a_{R,1}+\sum_{k=2}^K b_k+o_d(1)\right)\right\}
    \end{align*}

    Now for a fixed $\{S_k\}$ satisfying (2) and (3), we  union bound over all  $(T+1)^K$ choices of $\{t_k\}$ to get an upper bound on the probability that (1) is satisfied.  Let $E_A$ be the number of edges $e$ in $K_{n,n}$ such that $\exists k \in [K]$ such both endpoints of $e$ lie within $S_k$. For a fixed $\{S_k\}$ and  $\{t_k\}$, condition (1) is satisfied iff a certain (at least) set of $E_A$ independent non-edges occur in the sampling of $\{A^{t_k}\}$, and this happens with probability at most $\left(1-\frac{d}{n}\right)^{E_A} \leq \exp(-dE_A/n)$. In addition, we have,
    \begin{align*}
        E_A &\geq \phi^2a_{L,1}a_{R,1}+\phi^2\sum_{k=2}^Ka_{L,k}a_{R,k}-(a_{L,k}-b_{L,k})(a_{R,k}-b_k+b_{L,k}) \\
        &= \phi^2a_{L,1}a_{R,1} +\phi^2\sum_{k=2}^K b_{L,k}^2+b_{L,k}(a_{R,k}-a_{L,k}-b_{k})+a_{L,k}b_{k} \\
        &\geq \phi^2a_{L,1}a_{R,1} +\phi^2\sum_{k=2}^K b_{L,k}(a_{R,k}-a_{L,k}-b_{k})+a_{L,k}b_{k} 
    \end{align*} 
    All together we have,
    \begin{align*}
        \E[N] &\leq \frac{n^{4K}}{\varepsilon^{2K}}(T+1)^K 
        \sup_{a_{L,k}, a_{R,k}, b_{k}, b_{L,k}}
        \exp\left\{\phi \log d\left(a_{L,1}+a_{R,1}+\sum_{k=2}^K b_k+o_d(1)\right)\right\}  
        \exp\left(-\frac{d}{n} E_A\right) \\
        &\leq \frac{n^{4K}}{\varepsilon^{2K}}(T+1)^K 
        \sup_{a_{L,1}, a_{R,1}} 
        \exp\big\{\phi \log d\left(a_{L,1}+a_{R,1}-a_{L,1}a_{R,1}\right)\big\} \times \\
        &~~~~~~~~~~\sup_{a_{L,k}, a_{R,k}, b_{k}, b_{L,k}}\exp\left\{\phi \log d\left(\sum_{k=2}^K b_k-b_{L,k}(a_{R,k}-a_{L,k}-b_{k})-a_{L,k}b_{k} +o_d(1)\right)\right\} \\
        &\leq \frac{n^{4K}}{\varepsilon^{2K}}(T+1)^K
        \sup_{a_{L,k}, a_{R,k}, b_{k}, b_{L,k}}\exp\left\{\phi \log d\left(1+\sum_{k=2}^K b_k-b_{L,k}(a_{R,k}-a_{L,k}-b_{k})-a_{L,k}b_{k} +o_d(1)\right)\right\}
    \end{align*}
    We used $\sup_{(x,y) \in (1,\infty)\times(1,\infty)} x+y-xy \leq 1$ for the last inequality. 
    
    If $a_{R,k}-a_{L,k}\geq b_k$, then 
    \begin{align*}
        \E[N] &\leq \frac{n^{4K}}{\varepsilon^{2K}}(T+1)^K \sup_{a_{L,k}, a_{R,k}, b_{k}}\exp\left\{\phi \log d\left(1-\sum_{k=2}^K b_k(a_{L,k}-1) +o_d(1)\right)\right\} \\
        &\leq \frac{n^{4K}}{\varepsilon^{2K}}(T+1)^K \exp\left\{\phi \log d\left(1-\sum_{k=2}^K \frac{\varepsilon}{4}(1+\varepsilon-1)+o_d(1)\right)\right\} \\
        &\leq \frac{n^{4K}}{\varepsilon^{2K}}(T+1)^K \exp\left\{\phi \log d\left(1-\sum_{k=2}^K \frac{\varepsilon^2}{4}+o_d(1)\right)\right\} \\
        &\leq \frac{n^{4K}}{\varepsilon^{2K}}(T+1)^K 
        \exp\left\{\phi \log d\left(1 -\frac{9}{\varepsilon^2} \cdot\frac{\varepsilon^2}{4}+o_d(1)\right)\right\} \\
        & \leq \frac{n^{4K}}{\varepsilon^{2K}}(T+1)^K  \exp\left\{\phi \log d\left(-\frac{5}{4}+o_d(1)\right)\right\} = \exp(-\Omega(n)).
    \end{align*}

    If $a_{R,k}-a_{L,k}< b_k$, then
    \begin{align*}
        \E[N] &\leq \frac{n^{4K}}{\varepsilon^{2K}}(T+1)^K
        \sup_{a_{L,k}, a_{R,k}, b_{k}, b_{L,k}}\exp\left\{\phi \log d\left(1+\sum_{k=2}^K b_k+b_{L,k}(b_{k}+a_{L,k}-a_{R,k})-a_{L,k}b_{k} +o_d(1)\right)\right\} \\
        &\leq \frac{n^{4K}}{\varepsilon^{2K}}(T+1)^K
        \sup_{a_{L,k}, a_{R,k}, b_{k}, b_{L,k}}\exp\left\{\phi \log d\left(1+\sum_{k=2}^K b_k+b_{k}(b_{k}+a_{L,k}-a_{R,k})-a_{L,k}b_{k} +o_d(1)\right)\right\} \\
        &\leq \frac{n^{4K}}{\varepsilon^{2K}}(T+1)^K
        \sup_{a_{L,k}, a_{R,k}, b_{k}, b_{L,k}}\exp\left\{\phi \log d\left(1-\sum_{k=2}^K b_{k}(a_{R,k}-b_k-1)+o_d(1)\right)\right\} \\
        &\leq \frac{n^{4K}}{\varepsilon^{2K}}(T+1)^K
        \sup_{a_{L,k}, a_{R,k}, b_{k}, b_{L,k}}\exp\left\{\phi \log d\left(1-\sum_{k=2}^K \frac{\varepsilon}{4}(1+\varepsilon-\varepsilon/2-1)+o_d(1)\right)\right\} \\
        &\leq \frac{n^{4K}}{\varepsilon^{2K}}(T+1)^K \exp\left\{\phi \log d\left(1-\sum_{k=2}^K \frac{\varepsilon^2}{8}+o_d(1)\right)\right\} \\
        &\leq \frac{n^{4K}}{\varepsilon^{2K}}(T+1)^K 
        \exp\left\{\phi \log d\left(1 -\frac{9}{\varepsilon^2} \cdot\frac{\varepsilon^2}{8}+o_d(1)\right)\right\} \\
        & \leq \frac{n^{4K}}{\varepsilon^{2K}}(T+1)^K  \exp\left\{\phi \log d\left(-\frac{1}{8}+o_d(1)\right)\right\} = \exp(-\Omega(n)).  \qedhere
    \end{align*}
\end{proof}

To end this section, we prove the following lemma which states that no independent set of $\Gdb$ has large intersection with a fixed set of vertices on both sides (compare to~\cite[Lemma 2.4]{wein2022optimal}).

\begin{lemma} \label{Lemma:Smallintersection}
    Fix  $\varepsilon>0$, $p>0$. Let $d = d(\varepsilon,p)>0$ be large enough, and fix any $S \subseteq [2n]$ with $|S \cap L|, |S \cap R| \leq p\frac{\log d}{d} n$. 
 Then with probability $1-\exp(-\Omega(n))$ there is no independent set $S'$ in $\Gdb$ such that $|S' \cap S \cap L| \geq \varepsilon \frac{\log d}{d} n$ and $|S' \cap S \cap R| \geq \varepsilon \frac{\log d}{d} n$.
\end{lemma}

\begin{proof}
    We use the first-moment method, and as in Proposition \ref{Prop:forbidden} let $\phi = \frac{\log d}{d} n$. Let $N_{\varepsilon,p}$ be the number of subsets $U \subseteq S$ such that $|U \cap L| = |U \cap R| = \lceil \varepsilon \phi\rceil := b\phi$ and $U$ is an independent set in $\Gdb$. It suffices to show that $N_{\varepsilon,p}=0$ with high probability. We have 
    \begin{align*}
        \mathbb{E}\left[N_{\varepsilon,p}\right] &= {|S \cap L| \choose b\phi}{|S \cap R| \choose b\phi} \left(1-\frac{d}{n}\right)^{(b\phi)^2} \\
        &\leq \left(\frac{ep}{b}\right)^{2b\phi} \exp \left\{-\frac{d}{n}\cdot b^2\phi^2\right\} \\
        &= \exp \left\{2b\phi \log\left(\frac{ep}{b}\right)-\frac{d}{n}\cdot b^2\phi^2\right\} \\
        & = \exp \left\{\phi \log d(-b^2+o_d(1)\right\}  ~~~ \text{using } b \in [\varepsilon,p]\\
        & = \exp\{-\Omega(n)\}
    \end{align*}
    for large enough $d$. Then we apply  Markov's Inequality. 
\end{proof}

\subsection{$(D, \Gamma, c)$-stable algorithms}\label{Subsec:D_stable}

We write $A \sim \Gdb$ to mean a random vector (or random bi-adjacency matrix) $A \in \{0,1\}^m$ , $m=n^2$, with i.i.d \text{Ber}($d/n$) entries. We consider the hypercube $Q_m$ with vertex set $\{0,1\}^m$ and an edge $(a,a')$ if and only if $a$ differs from $a'$ in exactly one coordinate. The interpolation path (of Definition \ref{Def:interpolationpath}) can be viewed as a random walk on $Q_m$. Now, proceeeding as in~\cite[Sec 2.3]{wein2022optimal}, we define $(D, \Gamma, c)$-stable algorithms with respect to this random walk.

\begin{Definition}
    Let $f : \{0,1\}^m \to \R^{2n}$ and $c >0$. An edge $(a,a')$ of the hypercube $Q_m$ is $c$-bad for $f$ if
    \[
        \|f(a)-f(a')\|^2 \geq c \cdot \mathbb{E}_{A\sim \Gdb}\left[\|f(A)\|^2\right]
    \]
\end{Definition}

\begin{Definition}
    Let $\Gamma \in \N$ and $c>0$. Let $A^{(0)}, \cdots, A^{(T)}$ be the interpolation path from Definition \ref{Def:interpolationpath} of length $T = \Gamma m$. A function $f : \{0,1\}^m \to \R^{2n}$ is said to be $(D,\Gamma,c)$-stable if
    \[
        \mathbb{P}\Big[\text{no edge of } A^{(0)}, \cdots, A^{(T)} \text{ is c-bad for f }\Big] \geq \left(\frac{d}{n}\right)^{4\Gamma D/c}.
    \]
\end{Definition}

That means, with some (quantitatively) non-trivial probability, the interpolation path has no bad edges for $f$, and hence the value of $f$ will not vary much during this walk. We show that bipartite local algorithms and low-degree algorithms for $\Gdb$ are in the class of $(D,\Gamma,c)$-stable algorithms for an appropriate choice of parameters. 

\begin{lemma} \label{Lemma:D_stable_local}
    Fix   $\Gamma \in \N$ and $c>0$, $s>0$, and $d>0$.  Then for $n$ large enough, any $s$-local bipartite algorithm for $\Gdb$ is $(1,\Gamma,c)$-stable.
\end{lemma}

\begin{proof}
    Let $g = (g_{\ell}, g_r)$ be a pair compatible $s$-local functions represented as a vector valued function.  We can assume that $\max \{ \E g_{\ell}(\mathbf T, \mathbf x) , \E g_{r}(\mathbf T, \mathbf x)  \}>0$ (where $\mathbf T$ is the Pois($d$) Galton-Watson tree) since otherwise the result is immediate; this implies that $ \mathbb{E}_{A\sim \Gdb}\left[\|g(A)\|^2\right] = \Omega(n)$.  Let $(a,a')$ be one edge of the hypercube $Q_m$. By definition, $a$ and $a'$ differ in exactly one coordinate and let $e = uv$, ($u \in L$, $v \in R$) be the corresponding edge. Since the output of $g_{\ell}$ only depends on the structure of $s$-neighbor of a vertex, then $g_{\ell}(a_L)$ and  $g_{\ell}(a'_L)$ differ in, at most, all vertices within the $s$-neighborhood of $u$,  which we denote $N_s(A,u)$. The same is true for $g_r$ and $v$. Thus, $\|g(a)-g(a')\|^2 \leq |N_{s}(A,u)|+|N_{s}(A,v)|$. 
    
    We now use a tail bound on $||N_{s}(A,v)|$ adapted from~\cite{wein2022optimal}. Let $A \sim G(n,d/n)$ and fix any vertex $u \in [n]$. Let $m_i$ be the number of the number of vertices whose distance in $A$ is exactly $i$ from $u$. 
    By conditioning on $m_i$, $m_{i+1}$ is stochastically dominated by $\text{Bin}(m_in, d/n)$. By a Chernoff and union bound,  the following tail bound for $|N_s(A,u)|$, where $A \sim G(n,d/n)$, is proved in~\cite{wein2022optimal}:
    \[
        \mathbb{P}\big[|N_s(A,u)|\geq t\big] \leq s\exp\left\{-\Omega\left(t^{\frac{1}{s}}\right)\right\}.
    \]
    The same stochastic domination  argument applies for the bipartite model $A \sim \Gdb$, thus we have the same tail bound for $|N_s(A,u)|$ where $A \sim \Gdb$.
    It follows that the probability of $\|g(a)-g(a')\|^2 \geq c \cdot \mathbb{E}_{A\sim \Gdb}\left[\|g(A)\|^2\right]$ is  at most $s\exp\{-\Omega(n^{1/s})\}$ since $\mathbb{E}_{A\sim \Gdb}\left[\|g(A)\|^2\right] = \Omega(n)$. By a union bound, for large enough $n$, 
    \[
        \mathbb{P}\Big[\text{no edge of } A^{(0)}, \cdots, A^{(T)} \text{ is c-bad for f }\Big] 
        \geq 1-T\cdot s\exp\{-\Omega(n^{1/s})\} 
        = 1-\Gamma n^2 s\exp\{-\Omega(n^{1/s})\}
        \geq \left(\frac{d}{n}\right)^{4\Gamma/c}. \qedhere
    \]
\end{proof}
    \begin{lemma}              \label{Lemma:D_stable_lowdeg}
        For any constants $\Gamma \in \N$ and $c>0$, any degree-D function $f$ is $(D,\Gamma,c)$-stable.
    \end{lemma}
 
  \cite[Lemma 2.7 and Lemma 2.8]{wein2022optimal} state that degree-$D$ polynomials are $D$-stable when $A \sim G(n,d/n)$. The same proof is easily adapted to our case of $A \sim \Gdb$.

\subsection{Proof of the impossibility results}
\label{secMainProof}

In this section, we prove the impossibility results of Theorem \ref{Thm:balhalfopt} and Theorem \ref{Thm:imposs_achieve_lowdeg}. We state the following proposition  which shows that no $(D,\Gamma,c)$-stable algorithm can find an independent set of density at least $(1+\varepsilon) \frac{\log d}{d}$ in both $L$ and $R$, in the sense of Definition~\ref{Def:f_Optimize_ind}. We then apply this proposition to local and low-degree algorithms. 

\begin{Proposition} \label{Prop:imposibility_D_stable}
    For any $\varepsilon>0$, there exists $d_0>0$ such that for any $d \geq d_0$ there exists $n_0 >0$, $\eta >0$ $C_1>0$ and $C_2>0$ such that for any $n \geq n_0$, $\xi \geq 1$, $1 \leq D \leq \frac{C_1 n}{\xi \log n}$ and $\delta \leq \exp(-C_2 \xi D \log n)$, if $\min \{k_{\ell},k_r\} \geq (1+\varepsilon)\frac{\log d}{d} n$, there is no $\left(D, \lceil \frac{9}{\varepsilon^2}\rceil, \frac{\varepsilon}{64 \xi(1+\varepsilon)}\right)$-stable function that $(k_{\ell}, k_r, \delta, \xi, \eta\big)$-optimizes the independent set problem in $\Gdb$.
\end{Proposition}

\begin{proof}
    Fix $\epsilon >0$, let $K = \lceil \frac{9}{\varepsilon^2}\rceil+1$, $T=(K-1)m$, $c = \frac{\varepsilon}{64 \xi(1+\varepsilon)}$ and $\eta = \frac{\varepsilon}{16} \frac{\log d}{d}$. Let $d_0 = d_0(\varepsilon)$ large enough that Lemma \ref{Lemma:balancex+y>xy}, Proposition \ref{Prop:forbidden} and Lemma \ref{Lemma:Smallintersection} hold true. Assume such an $(D,K-1,c)$-stable algorithm $f$ exists. Sample the interpolation path $A^{(0)}, \cdots, A^{(T)}$ as in Definition \ref{Def:interpolationpath}, and let $V_{t} = V^{\eta}_f(A^{(t)})$ be the resulting independent set. We constructed a sequence of sets $S_1, \cdots, S_K$ based on $\{V_t\}$ as follows. Let $S_1 = V_0$, and for $k \geq 2$, let $S_k$ be the first $V_t$ such that $|V_t \setminus \cup_{l < k}S_l| \geq \frac{\varepsilon}{4} \frac{\log d}{d} n$. If no such $t$ exists then the process fails. We define the following three events:

    \begin{enumerate}
        \item $|V_t \cap L| \geq (1+\varepsilon) \frac{\log d}{d}n $ and $|V_t \cap R| \geq (1+\varepsilon) \frac{\log d}{d}n $ for all $t \in [T]$, and the process of constructing $S_1, \cdots, S_K$ succeeds.

        \item No edge of the interpolation path is $c$-bad for $f$.

        \item The forbidden structure of Proposition \ref{Prop:forbidden} does not exist.
    \end{enumerate}

    Now we state the first lemma, which says if $(1)$ and $(2)$ happen, then the forbidden structure will be produced, which implies that $(1)$, $(2)$ and $(3)$ can not happen at the same time. 

    \begin{lemma} \label{Lemma:1,2->not 3}
        If (1) and (2) both happen, then the  sequence of sets $S_1, \cdots, S_k$ satisfies the properties of the forbidden structure of Proposition \ref{Prop:forbidden}.
    \end{lemma}

    \begin{proof}
        First we show that $|V_t \triangle V_{t-1}| \leq \frac{\varepsilon}{4} \frac{\log d}{d} n$. From (1) we know that the failure event making $V^{\eta}_f = \varnothing$ does not occur for all $A^{(t)}$, $t \in [T]$. By definition \ref{Def:V_eta_f}, $i \in V_t \triangle V_{t-1}$ if two events occur:
        \begin{itemize}
            \item $i \in (I \setminus \Tilde{I}) \cup J$ for either $V^{\eta}_f(A^{(t)})$ or $V^{\eta}_f(A^{(t-1)})$, or
            \item One of $f_i(A^{(t)})$ and $f_i(A^{(t-1)})$ is $\geq 1$ and the other is $\leq \frac{1}{2}$. 
        \end{itemize}

        Notice that case 1 happens for  at most $2\eta n$ coordinates given $|(I \setminus \Tilde{I}) \cup J| \leq \eta n$ in definition \ref{Def:V_eta_f}. Hence there are at least $|V_t \triangle V_{t-1}|-2\eta n$ coordinates $i \in [2n]$ for which $|f_i(A^{(t)})-f_i(A^{(t-1)})| \geq \frac{1}{2}$. Together with event (2), no $c$-bad edges for $f$, it gives
        \begin{align*}
            \frac{1}{4}\big(|V_t \triangle V_{t-1}|-2\eta n\big) \leq \big\|f(A^{(t)})-f(A^{(t-1)})\big\|^2 \leq c \cdot \mathbb{E}_{A \sim \Gdb}\big[\|f(A)\|^2\big]
        \end{align*}
        By assumption that $f$ $(k_{\ell}, k_r, \delta, \xi, \eta\big)$-optimizes the independent set problem, we have
        \begin{align*}
            \mathbb{E}_{A \sim \Gdb}\big[\|f(A)\|^2\big] \leq 2\xi(1+\varepsilon)\frac{\log d}{d}n
        \end{align*}
        Put two inequalities together, we have
        \begin{align*}
            |V_t \triangle V_{t-1}| \leq 8c\xi(1+\varepsilon)\frac{\log d}{d}n+2\eta n 
            = 8 \cdot \frac{\varepsilon}{64 \xi (1+\varepsilon)} \cdot \xi (1+\varepsilon)\frac{\log d}{d}n+ 2 \cdot \frac{\varepsilon}{16} \frac{\log d}{d}n 
            = \frac{\varepsilon}{4}\frac {\log d}{d}n.
        \end{align*}
        Recall that $S_k$ is the first $V_t$ for which $|V_t \setminus \cup_{l \leq k}S_l| \geq \frac{\varepsilon}{4}\frac {\log d}{d}n$. That means $|V_{t-1} \setminus \cup_{l \leq k}S_l| \leq \frac{\varepsilon}{4}\frac {\log d}{d}n$ and we have
        \begin{align*}
            |S_k \setminus \cup_{l \leq k}S_l| 
            = |V_t \setminus \cup_{l \leq k}S_l| 
            = |(V_t \cap V_{t-1}) \setminus \cup_{l \leq k}S_l|
            + |(V_t \setminus V_{t-1}) \setminus \cup_{l \leq k}S_l| 
            \leq |V_{t-1} \setminus \cup_{l \leq k}S_l| + |V_t \triangle V_{t-1}| \leq \frac{\varepsilon}{2}\cdot\frac {\log d}{d}n.
        \end{align*}
        Thus, $|S_k \setminus \cup_{l \leq k}S_l| \in [\frac{\varepsilon}{4}\frac {\log d}{d}n, \frac{\varepsilon}{2}\frac {\log d}{d}n]$. for all $2 \leq k \leq K$. Combining this with event (1) and $S_k$ is an independent set in $A^{(j)}$ for some $j$, we conclude that $S_1, \cdots, S_k$ satisfies the properties of the forbidden structure in Proposition \ref{Prop:forbidden}.
    \end{proof}

    Now we state the second lemma which shows that with positive probability that events 1,2 and 3 occur simultaneously. 
    \begin{lemma} \label{Lemma:occursimul}
        There exist constants $C_1, C_2 >0$  such that events 1,2 and 3 as stated in Proposition \ref{Prop:imposibility_D_stable} occur simultaneously with positive probability. 
    \end{lemma}
    \begin{proof}
        We first  bound  the probability that event 1 occurs. Since $f$ $(k_{\ell}, k_r, \delta, \xi, \eta\big)$-optimizes the independent set problem, we have
        \[            \mathbb{P}_{A\sim\Gdb}\left[\left|V^\eta_f(A) \cap L\right| \geq (1+\varepsilon)\frac{\log d}{d} n \text{ and }\left|V^\eta_f(A) \cap R\right| \geq (1+\varepsilon)\frac{\log d}{d} n\right] \geq 1-\delta.
        \]
        Combining with Lemma \ref{Lemma:balancex+y>xy}, we have $(1+\varepsilon)\frac{\log d}{d} n \leq |V_t \cap L|,  |V_t \cap R| \leq (1+1/\varepsilon)\frac{\log d}{d} n$ with probability at least $1-\delta-\exp\{-\Omega(n)\}$. Now suppose that for some $0 \leq T' \leq T-m$, $A^{(0)}, \cdots, A^{(T')}$ have been sampled and $S_1 = V_0, \cdots, S_{K'} = V_{t_{k'}}$ have been successfully selected. As mentioned previously, $A^{(T'+m)}$ is independent from $A^{(0)}, \cdots, A^{(T')}$. So, provided $|S_k| \leq (1+1/\varepsilon)\frac{\log d}{d}n$ for $1 \leq k \leq K'$, using lemma \ref{Lemma:Smallintersection} with $S = \cup_{k \leq K'}S_k$ and $p = (1+1/\varepsilon)K'$ gives $|V_{T'+m} \cap (\cup_{k \leq K'}S_k) \cap L| \leq \varepsilon \frac{\log d}{d} n$ or $|V_{T'+m} \cap (\cup_{k \leq K'}S_k) \cap R| \leq \varepsilon \frac{\log d}{d} n$ with probability $1-\exp\{-\Omega(n)\}$. Without loss of generality, assume the first inequality holds, we have 
        \[
        |V_{T'+m} \setminus \cup_{k \leq K'}S_k| \geq |(V_{T'+m} \setminus \cup_{k \leq K'}S_k) \cap L| = |V_{T'+m} \cap L| - |V_{T'+m} \cap (\cup_{k \leq K'}S_k) \cap L| \geq \varepsilon\frac{\log d}{d}n \geq \frac{\varepsilon}{4}\frac{\log d}{d}n.
        \]
        This implies that we can find $S_k =  V_t$ for some $t \leq T'+m$ and thus by induction the process succeeds by time step $T = (K-1)m$. Using a union bound over $t$, we have event 1 holds with probability at least $1-\delta(T+1)-\exp\{-\Omega(n)\}$.

        Since $f$ is assume to be $(D,K-1, c)$-stable, event 2 holds with probability at least $(d/n)^{4(K-1)D/c}$. By Proposition \ref{Prop:forbidden}, event 3 holds with probability $1-\exp\{-\Omega(n)\}$. We claim that it suffices to have 
        \begin{align} \label{Ine:C_1andC_2}
            (d/n)^{4(K-1)D/c}>2\exp\{-Cn\} 
            \text{ and }
            (d/n)^{4(K-1)D/c} >  2\delta Kn^2
        \end{align}
        for some constant $C=C(\varepsilon,n)$ to conclude that all three events happen simultaneously with non-zero probability, since (\ref{Ine:C_1andC_2}) implies
        \[
         (d/n)^{4(K-1)D/c} >  \exp\{-Cn\}+\delta Kn^2> \exp\{-\Omega(n)\}+\delta Km> \exp\{-\Omega(n)\}+\delta(T-1)
        \]
        For $d>1$, the first inequality in (\ref{Ine:C_1andC_2}) is implied by $D \leq (Cn - \log 2)\frac{c}{4(K-1)\log n} = (Cn - \log 2)\frac{\varepsilon}{64\xi(1+\varepsilon)4(K-1)\log n}$. And this is implied by $D< \frac{C_1n}{\xi \log n}$, for large enough $n$ and some constant $C_1 = C_1(\varepsilon, d)>0$. The second inequality is implied by $\delta \leq \exp \{-\frac{4(K-1)D}{c}\log n-2\log n-\log (2K)\}$. For large enough $n$, given $\xi,D \geq 1$, there exists another constant $C_2 = C_2(\varepsilon, d)>0$ making it implied by $\delta \leq \exp\{-C_2\xi D \log n\}$. 
        
    \end{proof}
    By Lemma \ref{Lemma:1,2->not 3} and \ref{Lemma:occursimul}, we obtain a contradiction by assuming such an $f$ exists. Hence the proposition follows. 
\end{proof}

Now we proof the negative results of Theorem \ref{Thm:balhalfopt} and \ref{Thm:imposs_achieve_lowdeg} based on Proposition \ref{Prop:imposibility_D_stable}. 
 
\begin{proof}[Proof of Theorem \ref{Thm:balhalfopt}, upper bound] 
     Fix $\varepsilon>0$ and let $ d\ge d_0$ as in Proposition~\ref{Prop:imposibility_D_stable}. Assume that for some $s>0$ there exists a pair of  $s$-local compatible functions $(g_{\ell}, g_{r})$ with $\alpha_{d,\frac{1}{2}}((g_{\ell}, g_{r})) \ge (1+\varepsilon)\frac{\log d}{d}$. By Lemma \ref{Lemma:GamValueRandom}, the random independent set obtained by applying $g_{\ell}, g_{r}$ to $\Gdb$ will have size $(1+\varepsilon)\frac{\log d}{d}n$ on both sides whp. For $\delta, \xi, \eta$ chosen as in Proposition \ref{Prop:imposibility_D_stable}, that also means, by Lemma \ref{Lemma:D_stable_local}, there exists a $\left(1, \lceil \frac{9}{\varepsilon^2}\rceil, \frac{\varepsilon}{64 \xi(1+\varepsilon)}\right)$-stable function $\left((1+\varepsilon)\frac{\log d}{d}n, (1+\varepsilon)\frac{\log d}{d}n, \delta, \xi, \eta\right)$-optimizes the independent set problem in $\Gdb$, a contradiction to Proposition~\ref{Prop:imposibility_D_stable}.
\end{proof}

 \begin{proof}[Proof of Theorem \ref{Thm:imposs_achieve_lowdeg}, part 2]
     Combining  Lemma \ref{Lemma:D_stable_lowdeg} and Proposition \ref{Prop:imposibility_D_stable} gives the second part of Theorem \ref{Thm:imposs_achieve_lowdeg}.
 \end{proof}

\section{Existence of large $\gamma$-balanced independent sets}
\label{Sec:Exist}

Here we prove the existential results of Theorem~\ref{Thm:balexist} in two steps. First we show the upper bound by the first-moment method,  and then use martingale inequalities and  the second-moment method for the lower bound, following~\cite{frieze1990independence}. 

Define $f(c,d) = c \cdot \frac{\log d}{d} $, which we denote just by $f$ for convenience. 

To prove the upper bound, we need to show that for any $c>1/2\gamma(1-\gamma)$ and large enough $d$,
 \begin{align*}
     \mathbb{P}[\Xg \leq 2n f] = 1-o(1)  \text{ as }n \to \infty \,.
 \end{align*}

 Let $Y_f$ be the number of $\gamma$-balanced independent set of size of $2nf$, then by Markov's inequality we have 
\begin{align*}
    \mathbb{P}\big[\Xg \ge 2nf\big] 
    = \mathbb{E}[Y_f\neq 0] \leq \mathbb{E}[Y_f] = {n \choose 2\gamma n f}{ n \choose 2(1-\gamma)nf }\Big(1-\frac{d}{n}\Big)^{4\gamma(1-\gamma)(nf)^2}    
\end{align*}
Now we want show that $\lim_{n \to \infty}\mathbb{P}[\Xg > 2nf]$=0, and to do this it suffices to show that $\limsup_{n \to \infty} n^{-1} \log( \mathbb{E}[Y_f]) \leq -\epsilon$ for some $\epsilon > 0$. Using Stirling's approximation and the fact that $\log(1-f) = -f (1+o_d(1))$ as $d \to \infty$, we have
\begin{align*}
    \limsup_{n \to \infty} \frac{1}{2n}\log\big( \mathbb{E}[Y_f]\big) 
    &=
    \limsup_{n \to \infty} \frac{1}{2n} \Bigg\{2\log \big(n!\big)
    -\log\Big(( 2\gamma n f)!\Big)-\log\Big(\big(n - 2\gamma n f\big)!\Big)\\
    &- \log\Big(\big( 2(1-\gamma) n f\big)!\Big) - \log\Big((n - 2(1-\gamma) n f)!\Big)
    +4\gamma(1-\gamma)(nf)^2\log \Big(1-\frac{d}{n}\Big)\Bigg\} \\
    &= \limsup_{n \to \infty} \frac{1}{2n} \Bigg\{- 2n\gamma f\log(2\gamma f) - n(1-2\gamma f)\log(1-2\gamma f)-2n(1-\gamma)f \log\big(2(1-\gamma)f\big)\\ 
    & - n\big(1-2(1-\gamma)f\big) \log\big(1-2(1-\gamma)f\big)
    - 4\gamma(1-\gamma)(nf)^2\frac{d}{n}
    \Bigg\}\\
    &= -\gamma f\log(2\gamma f) - \left(\frac{1}{2}-\gamma f\right)\log(1-2\gamma f)-(1-\gamma)f \log\big(2(1-\gamma)f\big)\\
    &-\left(\frac{1}{2}-(1-\gamma) f\right)\log\big(1-2(1-\gamma)f\big)
    - 2\gamma(1-\gamma)f^2d \,.
\end{align*}
Note that the  leading term is $c\big(1-2\gamma(1-\gamma)c\big) \frac{\log^{2}d}{d}$, which is negative for $c>1/2\gamma(1-\gamma)$ and hence the upper bound follows.

Next we prove the lower bound by showing that for any $c<1/2\gamma(1-\gamma)$ and large enough $d$,  $\mathbb{P} [\Xg \ge 2nc \cdot \frac{ \log d}{d}] = 1-o(1)$. We show this by using Azuma's inequality and the second-moment method  as in~\cite{frieze1990independence}. The idea is to show the high probability existence of an independent set from a special class of balanced independent sets, balanced $P$-independent sets (defined below), which implies the lower bound for $\Xg$. 

Fix $\gamma \in (0,\frac{1}{2}]$.  We denote the vertices in $L$ and $R$ as $L = \{1_L, \cdots, n_L\}$, $R = \{1_R, \cdots, n_R\}$ and let $m = \Big[\frac{d}{(\log d)^2} \Big]$, $n' = \lfloor n/m \rfloor$. For each $i \in [n']$, let $P_{i,L} = \{((i-1)m+1)_L, \cdots, (im)_L\}$ and $P_{i,R} = \{((i-1)m+1)_R, \cdots, (im)_R\}$. So we have $n'$ partition sets for both $L$ and $R$. We say a set $I$ is a $\gamma$-balanced $P$-independent set if $X$ is a $\gamma$-balanced independent set and $|X \cap P_{i,L}| \leq 1$, $|X \cap P_{i,R}| \leq 1$ for all $i \in [n']$. Let $\beta(G)$ denote the size of the largest $\gamma$-balanced $P$-independent subset of $G \sim \Gdb$ and $Z_k$ denote that number of $\gamma$-balanced P-independent subset of size $k$. 
 We will show the following, which established the desired lower bound on $X_\gamma$. 
    \begin{lemma}\label{lowerboundbetaGdb}
        For large enough $d$, 
        \begin{align*}
             \lim_{n \to \infty} \mathbb{P}\Big[\beta(\Gdb) \geq \frac{1}{\gamma(1-\gamma)}\frac{n}{d}(\log d - \log \log d-\log 2 +1-\varepsilon)\Big] = 1.
        \end{align*}
    \end{lemma}
To show Lemma \ref{lowerboundbetaGdb}, we establish two  lemmas similar to those in~ \cite{frieze1990independence}.
    \begin{lemma} \label{betaconcentration}
        Let $\overline{\beta} = \mathbb{E}\big[\beta(\Gdb)\big]$. Then
        \begin{align*}
            \mathbb{P}\Big[\big|\beta(\Gdb)-\overline{\beta}\big| \geq t\Big] \leq 2\exp\left\{-\frac{t^2 d}{4 (1+\frac{1}{\gamma})^2 (\log d)^2 n}\right\} \qquad \text{for } t>0.
        \end{align*}
    \end{lemma}

    \begin{lemma} \label{lowerboundZk}
        For fixed $\varepsilon \in (0,1)$, let $k = \frac{n}{\gamma(1-\gamma)d}(\log d - \log \log d-
    \log 2 +1 - \varepsilon/3)$, then
    \begin{align*}
        \mathbb{P}[Z_k>0] \geq \exp\left\{-\frac{2e^4}{\sqrt{\gamma(1-\gamma)}}\frac{(\log d)^{\frac{5}{2}}}{d^{\frac{3}{2}}}n\right\}.
    \end{align*}
    \end{lemma}

     Lemma \ref{lowerboundbetaGdb} follows by combining Lemmas \ref{betaconcentration} and \ref{lowerboundZk}. Indeed take $t = \frac{\epsilon n}{3\gamma(1-\gamma)d}$ in Lemma \ref{betaconcentration} and compare it to Lemma \ref{lowerboundZk}, we have $\overline{\beta} \geq k-t$. Apply $t = \frac{\varepsilon n}{3\gamma(1-\gamma)d}$ in Lemma \ref{betaconcentration} again, we get $\mathbb{P}\left[\beta(\Gdb) \leq k - 2t\right]$ is exponentially small as $n \to \infty$ which gives Lemma \ref{lowerboundbetaGdb}.

    To show Lemma \ref{betaconcentration}, we notice that $|\beta(G)-\beta(G')| \leq 1+\frac{1}{\gamma}$, where $G'$ is obtained from $G$ by changing some of the edges incident with the vertices in a single $P_{i,L}$ or $P_{i, R}$.  Then Azuma's  inequality  gives Lemma \ref{betaconcentration}.
    
    It remains to show Lemma \ref{lowerboundZk}. 
    \begin{proof}[Proof of Lemma \ref{lowerboundZk}]
    We use the inequality  $\mathbb{P}[Z_k >0] \geq \mathbb{E}[Z_k]^2/\mathbb{E}[Z^2_k]$. To compute the first and the second moment of $X_k$, 
    \begin{align*}
        \mathbb{E}[Z_k]= 
        {n' \choose \gamma k}
        m^{\gamma k} 
        {n' \choose (1-\gamma)k}
        m^{(1-\gamma)k}
        \left(1-\frac{d}{n}\right)^{\gamma(1-\gamma)k^2} 
        =  m^{k} {n' \choose \gamma k} 
        {n' \choose (1-\gamma)k}
        \left(1-\frac{d}{n}\right)^{\gamma(1-\gamma)k^2} 
    \end{align*}
    
    And \begin{align*}
        \mathbb{E}[Z_k^2]&\leq \mathbb{E}[Z_k]\sum_{l=0}^{k} m^{k-l} \sum_{i=0}^l{\gamma k \choose i}{n'-i \choose \gamma k-i}{(1-\gamma)k \choose l-i}{n'-(l-i) \choose (1-\gamma)k-(l-i)}
        \left(1-\frac{d}{n}\right)^{\gamma(1-\gamma)k^2-i(l-i)}
    \end{align*}

    Moreover, we have
    \begin{align*}
        \sum_{i=0}^{l}\frac{ {n'-i \choose \gamma k-i}
        {n'-(l-i) \choose (1-\gamma)k-(l-i)}}{{n' \choose \gamma k}{n' \choose (1-\gamma)k}} \leq  \sum_{i=0}^{l}\left(\frac{\gamma k}{n'} \cdot \frac{(1-\gamma) k}{n'}\right)^i \leq \left(\frac{\gamma(1-\gamma) k}{n'}\right)^l \cdot \frac{1}{k} 
        \leq \left(\frac{\gamma(1-\gamma) k}{n'}\right)^l \cdot \frac{1}{l}
    \end{align*}
Thus,
\begin{align}
    \frac{\mathbb{E}[X_k^2]}{\mathbb{E}[X_k]^2} 
    &\leq 
    \sum_{l=0}^{k} 
    \left(\frac{\gamma(1-\gamma) k}{n'm}\right)^l\frac{1}{l}
    \sum_{i=0}^l {\gamma k \choose i}{(1-\gamma)k \choose l-i}\left(1-\frac{d}{n}\right)^{-i(l-i)} \nonumber\\
    &\leq \sum_{l=0}^{k} 
    \left(\frac{\gamma(1-\gamma) k}{n'm}\right)^l\frac{1}{l}{k \choose l}
    \sum_{i=0}^l\left(1-\frac{d}{n}\right)^{-i(l-i)} \nonumber \\
    &\leq \sum_{l=0}^{k} 
    \left(\frac{\gamma(1-\gamma) k}{n'm}\right)^l\frac{1}{l}{k \choose l}\cdot l \left(1-\frac{d}{n}\right)^{-l^2/4}\nonumber \\
    &= \sum_{l=0}^{k} 
    \left(\frac{\gamma(1-\gamma) k}{n'm}\right)^l{k \choose l} \left(1-\frac{d}{n}\right)^{-l^2/4} \nonumber \\
    &\leq \exp\{4(\log d)^2\}
    \sum_{l=0}^{k} U_l \label{Inq:Exp_lessthan_U_l} \,,
\end{align}
where
\begin{align*}
    U_l = \left(\frac{\gamma(1-\gamma) k}{n'm}\right)^l{k \choose l}
     \exp\left\{\frac{l^2d}{4n}\right\} \,.
\end{align*}
The last inequality holds since 
\begin{align*}
    & ~~~~~~~ \left(1-\frac{d}{n}\right)^{l^2/4} \geq \exp\{-4(\log d)^2\} \exp \left \{-\frac{l^2d}{4n} \right \} \\
    &\Longleftrightarrow
    \frac{l^2}{4} \log \Big(1-\frac{d}{n}\Big)
    \geq
    -4(\log d)^2 -\frac{l^2d}{4n} \\
    &\Longleftrightarrow
    \log \Big(1-\frac{d}{n}\Big)
    \geq
    \frac{-16}{l^2}(\log d)^2-\frac{d}{n}, ~ ~ ~(\text{Notice that } \frac{d}{n} \to 0 \text{ gives } \log(1-\frac{d}{n}) \sim -\frac{d}{n}-\frac{d^2}{n^2})\\
    &\Longleftrightarrow
    -\frac{d}{n} -\frac{d^2}{n^2}\geq \frac{-16}{l^2}(\log d)^2-\frac{d}{n} \\
    &\Longleftrightarrow \frac{l^2}{n^2} \leq 16 \Big(\frac{\log d}{d}\Big)^2\\
    &\text{This is true due to the fact that } l \leq k < \frac{4n \log d}{d} \text{ for large enough } d.
\end{align*}
    
Similar to the upper bound of $u_l$ (up to a constant factor in the $\exp$ function) in~\cite{frieze1990independence}, we can upper bound $U_l$ as follows, 
\begin{align}
    U_l 
    &\leq  \left(\frac{\gamma(1-\gamma) k}{n'm}\right)^l \left(\frac{ke}{l}\right)^{l}
     \exp\left\{\frac{l^2d}{4n}\right\} \nonumber\\
     &\leq
      \left(\frac{\log d}{d}\cdot \frac{ke}{l}\cdot
     \exp\left\{\frac{ld}{4n}\right\}\right)^{l} \label{Inq:upperboundU_l}
\end{align}

\textit{Case 1: }$0 \leq l \leq 2\gamma(1-\gamma)k$.
Here $\exp\left\{\frac{ld}{4n} \right\} \leq \sqrt{d}$. And so, by (\ref{Inq:upperboundU_l}),
\begin{align}
    U_l &\leq \left(\frac{ek\log d}{l\sqrt{d}}\right)^{l} \nonumber \\
    &\leq \exp\left\{\frac{l \log d}{d}\right\}\nonumber \\
    &\leq  \exp\left\{\frac{(\log d)^2n }{\gamma(1-\gamma)d^{\frac{3}{2}}}\right\} \label{Inq:U_l_1}
\end{align}

\textit{Case 2: }$2\gamma(1-\gamma)k< l \leq \frac{4n}{d}(\log d-\log\log d-3)$.
By (\ref{Inq:upperboundU_l}),
\begin{align}
    U_l &\leq \left(\frac{4e\log d}{d} \cdot \exp\left\{\frac{ld}{4n}\right\}\right)^l \nonumber \\
    &\leq \left(\frac{4e\log d}{d} \cdot \frac{d}{e^3 \log d }\right)^l \nonumber \\
    &\leq 1  \label{Inq:U_l_2}
\end{align}

\textit{Case 3: }  $\frac{4n}{d}(\log d-\log\log d-3)<l\leq k$. Now,
\begin{align*}
    \frac{U_l}{U_{l+1}} &= \frac{(l+1)n'm}{(k-l)l}\exp\left\{-\frac{(2l+1)d}{4n}\right\} \\
    & \leq \frac{kn}{(k-l)l} \cdot \frac{e^6 (\log d)^4}{d^2}
\end{align*}

Hence,
\begin{align}
    U_l &\leq \frac{1}{(k-l)!l(l+1)\cdots(k-1)}\left(\frac{kne^6(\log d)^4}{d^2}\right)^{k-l}U_k \nonumber\\
    &\leq \frac{1}{((k-l)!)^2}\left(\frac{kne^6(\log d)^4}{d^2}\right)^{k-l} U_k \nonumber\\
    &\leq \left(\frac{kne^8(\log d)^4}{(k-l)^2d^2}\right)^{k-l} U_k \nonumber\\
    &\leq \exp \left\{2\left(\frac{kne^8(\log d)^4}{d^2}\right)^{1/2}\right\} U_k \nonumber\\
    &\leq \exp \left\{\frac{2e^4}{\sqrt{\gamma(1-\gamma)}}\frac{(\log d)^{\frac{5}{2}}}{d^\frac{3}{2}}\right\} U_k \label{Inq:U_l_3}
\end{align}

Finally, as shown in Inequality (8) in ~\cite{frieze1990independence} for some $\theta(d) \to 0$ as $d \to \infty$, we have 
\begin{align}
    U_k^{-1} = \left(\frac{n'm}{\gamma(1-\gamma)k}\exp\left\{-\frac{kd}{4n} \right\}\right)^k  = \Big[\big(1-\theta(d)\big)e^{\varepsilon/3}\Big]^k\geq 1. \label{Inq:U_l_4}
\end{align}

Now Lemma \ref{lowerboundZk} follows from (\ref{Inq:Exp_lessthan_U_l}) and (\ref{Inq:U_l_1})-(\ref{Inq:U_l_4}).   
\end{proof}

\bibliography{IndSetBib}
\bibliographystyle{plain}

\end{document}